\documentclass[12pt]{article}
\usepackage[font=small]{caption}
\usepackage[english]{babel}
\usepackage{float}
\usepackage{rotating, graphicx}
\usepackage{amsmath}
\usepackage{amssymb, amsthm}
\usepackage{array}
\usepackage{verbatim}
\usepackage[T1]{fontenc}   
\usepackage[round]{natbib}
\usepackage{setspace}
\usepackage{epstopdf}
\usepackage[table]{xcolor}    

\renewcommand{\arraystretch}{1.2}

\DeclareMathOperator{\E}{\mathsf{E}}   
\DeclareMathOperator{\Var}{\mathsf{Var}}   
\DeclareMathOperator{\Cov}{\mathsf{Cov}}   

\newcommand\independent{\protect\mathpalette{\protect\independenT}{\perp}}
\def\independenT#1#2{\mathrel{\rlap{$#1#2$}\mkern2mu{#1#2}}}

\newcolumntype{L}[1]{>{\raggedright\let\newline\\\arraybackslash\hspace{0pt}}m{#1}}
\newcolumntype{C}[1]{>{\centering\let\newline\\\arraybackslash\hspace{0pt}}m{#1}}
\newcolumntype{R}[1]{>{\raggedleft\let\newline\\\arraybackslash\hspace{0pt}}m{#1}}

\begin{document}

\newtheorem{theorem}{Theorem}[section]
\newtheorem{proposition}[theorem]{Proposition}
\newtheorem{example}{Example}[section]

\pagenumbering{arabic}

\title{Analysis of a five-factor capital market model}
\date{September, 2017}
\author{S{\o}ren Fiig Jarner and Michael Preisel \\[3mm] Quantitative Research \\ Danish Labour Market Supplementary Pension Fund (ATP)}
\maketitle

\begin{abstract}
In this paper we analyse the five-factor capital market model of \cite{munetal04}. The model features a Vasicek interest rate model, an equity index with mean-reverting excess return and an index for realized inflation with mean-reverting expectation. The primary aim of the analysis is to facilitate so-called exact simulation from the model on a set of discrete time points. It turns out that this can be achieved by sampling from a (degenerate) seven-dimensional normal distribution. We derive the distributional results necessary and describe how to overcome the rank deficiency of the variance-covariance matrix in practice.

The tradeable assets in the original model consist of cash, nominal bonds and stocks. We extend the investment universe to also include inflation bonds by deriving the arbitrage free break-even inflation (BEI) curve for a three-parameter specification of the two market prices of inflation risk. Finally, we provide a number of auxiliary results regarding the dynamics of constant-maturity nominal and inflation bond indices, the distribution of the stock index in nominal and real terms, and the distribution of the Sharpe ratio for individual assets and portfolios with an application to factor investing.
\end{abstract}

\bigskip


\bigskip

\thispagestyle{empty}
\newpage

\section{Introduction}
In this paper we analyse a slight variation of the continuous-time capital market model of \cite{munetal04}. The model consists of a nominal (short) rate, an equity index with mean-reverting risk premium, and an inflation index with mean-reverting rate of inflation. This gives a total of five state variables (factors), and we consequently refer to the model as a five-factor model. The model provides an analytically tractable, yet flexible, framework for joint modelling of nominal and real yields.

In \cite{munetal04} the model is used to find the optimal asset allocation strategy of a power utility investor. The investor can invest in cash, nominal bonds, and stocks and he seeks to optimize the utility of his real wealth on a given horizon, i.e.\ his nominal wealth divided by the value of the inflation index at the horizon. Under this model, investors with long horizons should invest a higher fraction of their wealth in stocks than investors with shorter horizons, and investors with higher risk aversion should have a higher bonds to stocks ratio than more risk seeking investors. Although these characteristics correspond to common financial recommendations of investment advisers, they do not arise in basic models of portfolio choice.\footnote{In particular, standard mean-variance two-fund separation results prescribe that all investors should hold the same mix of risky assets (same bonds to stocks ratio) and that the allocation to risky assets should remain constant regardless of the horizon.} In this respect, the model reconciles practical and theoretical recommendations.

In contrast to \cite{munetal04} the applications we have in mind are not optimization per se, but rather evaluation of different investment strategies, i.e.\ calculation of terminal wealth of different investment strategies in both real and nominal terms. In typical applications the investor makes periodic deposits to his (investment) account and the account is invested following a specific strategy. Due to the inflow of funds at discrete points in time it is not possible to obtain a closed form solution to terminal wealth. Even in the simplest case, terminal wealth takes the form of a sum of log-normal distributions, which can at best be approximated analytically. However, the (exact) distribution of the investor's wealth prior to his next deposit can be calculated and by successive simulations from these distributions it is possible to obtain samples from the terminal wealth distribution. The primary aim of the paper is to provide the formulas necessary to perform this so-called exact simulation of terminal wealth. As the name suggests exact simulation provides samples from the (exact) theoretical distribution, as opposed to approximate samples obtained e.g.\ from a discretized version of the diffusion processes. 

Simpler versions of the model have previously been analysed in \cite{jarpre16} and \cite{jorsli16} also with the aim of exact simulation. In \cite{jarpre16} a three-factor model for nominal rates and inflation (but without stocks) is analysed, and in \cite{jorsli16} a three-factor model for nominal rates and stocks (but without inflation) is analysed. The five-factor model can be seen as a combination of these two models. The model of \cite{marmil14} is also close in spirit, but strictly speaking neither a specialization nor a generalization of the present model.

The investable assets in \cite{munetal04} consist of cash, nominal bonds and stocks. However, if the investor is concerned about the real value of his (terminal) wealth it is also of interest to include the possibility for him to invest directly in real assets. With these applications in mind, we include a model for break-even inflation (BEI) and we derive the dynamics of inflation bond prices. The inclusion of inflation bonds as investable financial assets represents an extension to the original model of \cite{munetal04}. The term structure of nominal interest rates is of the form considered by \cite{vas77}, while the BEI curves are similar in structure but more involved due to two sources of randomness (rather than one) and the possibility of non-zero correlation between nominal rates and inflation.

The rest of the paper is organised as follows. In Section~\ref{sec:model} the dynamics of the model is specified and Section~\ref{sec:intrep} derives an integral representation of the state variables, which will be used throughout; Section~\ref{sec:exactsim} contains the results necessary for exact simulation, including easily verifiable conditions for the model being well-specified; Section~\ref{sec:bonds} adds break-even inflation to the model and derives the dynamics of nominal and inflation bond prices and indices; Section~\ref{sec:stockinfl} derives the distributions of the nominal and real value of stocks over time, which are useful for analytic investigation of model properties; Section~\ref{sec:portfolio} derives the Sharpe ratio distribution for a portfolio and gives an application to factor investing; Section~\ref{sec:concluding} offers a few concluding remarks; finally, Appendix~\ref{app:computation} contains  the computation of (part of) the variance-covariance matrix needed for the exact simulation scheme of Section~\ref{sec:exactsim}.

We aim at providing general formulas valid also in the limiting cases of vanishing volatility terms and vanishing mean reversion parameters, and for this reason most formulas are stated in terms of a number of auxiliary functions defined in Appendix~\ref{app:computation}.

\section{Model specification} \label{sec:model}
The aim of the model is to capture a number of stylized capital market features in a succinct and tractable way. In particular, correlated rates, equity returns and (price) inflation, a term-structure of interest rates and mean-reverting equity returns. The latter allows for a specification where the ``riskiness'' of stocks increases slower over time than under the standard Black-Scholes model, i.e.\ the mathematical equivalent of the intuition that stocks have a tendency to rebound from crashes. In Section~\ref{sec:bonds} we will also add break-even inflation to the model.

It is important to note that the dynamics specified below describe the so-called 'real world' dynamics of the state variables of the model (sometimes referred to as the $P$-dynamics as opposed to the $Q$-dynamics used for pricing). The price of nominal bonds (and later inflation bonds) can be derived from these dynamics when in addition we specify market prices of risk for the various risk sources.

\subsection{Short rate, inflation and stock returns}
We assume that the short (nominal) interest rate follows an Ornstein-Uhlenbeck process,
\begin{align}  \label{eq:shortrate}
   dr_t = \kappa(\bar{r} - r_t)dt + \sigma_r dW^r_t,
\end{align}
where $\bar{r}$ is the long-run mean of the short interest rate, $\kappa$ describes the degree of mean reversion, $\sigma_r$ is the interest rate volatility, and $W^r$ is a standard Brownian motion.

The stock index (total return index) is assumed to evolve according to the dynamics
\begin{align}  \label{eq:stockindex}
   \frac{dS_t}{S_t} = (r_t + x_t)dt + \sigma_S dW^S_t,
\end{align}
where $r_t$ is the short rate from (\ref{eq:shortrate}), $x_t$ is the time-varying expected excess return (or risk premium) from investing in stocks, $\sigma_S$ is the stock index volatility, and $W^S$ is a standard Brownian motion. We further assume that the excess return follows an Ornstein-Uhlenbeck process,
\begin{align}  \label{eq:excessreturn}
   dx_t = \alpha(\bar{x} - x_t)dt - \sigma_x dW^S_t,
\end{align}
where $\bar{x}$ denotes the long-run equity risk premium, $\alpha$ describes the degree of mean reversion towards this level, and $\sigma_x$ is the excess return volatility. Note that the stock index and the expected return processes are locally perfectly negatively correlated, i.e.\ a stock return above or below its expected value will ``cause'' a change in the (future) expected return in the opposite direction. This interaction induces a mean reversion in the stock returns over time.

We also introduce a price index process, $I$, interpreted as the nominal price of a real consumption good. We assume that $I$ evolves according to the dynamics
\begin{align}  \label{eq:inflindex}
   \frac{dI_t}{I_t} = \pi_t dt + \sigma_I dW^I_t,
\end{align}
with
\begin{align}   \label{eq:expectedinfl}
   d\pi_t = \beta(\bar{\pi} - \pi_t)dt + \sigma_\pi dW^{\pi}_t,
\end{align}
where $\pi_t$ is the expected rate of inflation, $\bar{\pi}$ is the long-run mean inflation rate, $\beta$ describes the degree of mean reversion, $\sigma_\pi$ is the volatility of the expected inflation rate, $\sigma_I$ is the volatility of the price index, and $W^I$ and $W^\pi$ are standard Brownian motions. The price index is thus influenced by both expected inflation and unexpected inflation shocks, where the expected inflation forms a persistent process while the shocks are uncorrelated.

In summary, the model has five stochastic quantities:
\begin{itemize}
  \item The short nominal rate ($r_t$)
  \item The equity index ($S_t$)
  \item The equity risk premium ($x_t$)
  \item The inflation index ($I_t$)
  \item The expected inflation ($\pi_t$)
\end{itemize}
Consequently, we refer to the model as a five-factor model. The model is driven by four Brownian motions and we assume joint normality of the quadruple $(W^r,W^S,W^\pi,W^I)$. To simplify matters slightly we will assume that unexpected inflation shocks are independent of the other driving Brownian motions, i.e.\ $W^I \independent (W^r,W^S,W^\pi)$, but apart from that we allow for a general dependency structure.

\subsection{Term structure of interest rates} \label{sec:TS}
To complete the model we also need to specify the evolution of the yield curve, i.e.\ the term structure of interest rates. Since the model is formulated in the language of financial mathematics it seems natural to conform with academic practice and require an arbitrage free model. It follows from arbitrage theory that the term structure is uniquely identified by the short rate dynamics under $P$, i.e.\ equation (\ref{eq:shortrate}), and the assumed market price of interest rate risk, $\lambda^r = \lambda^r(t,r_t)$, see e.g.\ Chapter 21 of \cite{bjo09}.

Assuming a linear market price of risk it can be shown that the term structure of interest rates is of the form considered by \cite{vas77}. In particular, the price at time $t$ of a zero-coupon bond maturing at time $T\geq t$ is given by
\begin{align}
    p_t(T) = \exp\left\{G(\Delta) - H(\Delta)r_t\right\},  \label{eq:VasZCBPrice}
\end{align}
with $\Delta = T-t$,
\begin{align}
  H(\Delta) & = \frac{1}{a}\left(1-\exp\{-a \Delta\}\right), \\
  G(\Delta) & = \left(b - \frac{\sigma_r^2}{2a^2}  \right)\left(H(\Delta)-\Delta \right) - \frac{\sigma_r^2}{4a}H^2(\Delta), \label{eq:VasGdef}
\end{align}
and where $a$ and $b$ are parameters controlling the slope and level of the yield curves.\footnote{The chosen parametrization corresponds to the market price of risk being of the form $\lambda^r_t = \{(a-\kappa)r_t + \kappa \bar{r} - a b\}/\sigma_r$. The price of a zero-coupon bond can be obtained by the usual risk-neutral valuation formula, where the risk-neutral interest rate dynamics are given by $dr_t = a(b-r_t)dt + \sigma_r d\bar{W}^r_t$. It is often assumed that the market price of interest rate risk is constant, in which case Ornstein-Uhlenbeck dynamics of $r$ under $P$ implies a Vasicek term structure. In general, however, the latter implication does not hold; although this is sometimes (incorrectly) stated.}

The continuously compounded zero-coupon yield for the period $[t,T]$, $r_t(T)$, is defined by the relation $p_t(T) = \exp\{-(T-t)r_t(T)\}$. It follows from (\ref{eq:VasZCBPrice})--(\ref{eq:VasGdef}) that
\begin{align}
  r_t(T) = \frac{1-\exp\{-a \Delta\}}{a \Delta}r_t + \left(b - \frac{\sigma_r^2}{2a^2}  \right)\left(1- \frac{H(\Delta)}{\Delta} \right) +  \frac{\sigma_r^2}{4a}\frac{H^2(\Delta)}{\Delta},   \label{eq:VasYield}
\end{align}
where $\Delta = T-t$. Since $H(\Delta)$ is uniformly bounded, it further follows from (\ref{eq:VasYield}) that all yield curves have the same asymptote
\begin{align}
  r_t(\infty) = \lim_{T\to\infty}r_t(T) =  b - \frac{\sigma_r^2}{2a^2}.
\end{align}
It is easy to show that this is also the asymptotic value of the forward rate, $f_t(T)$, defined by $f_t(T) = - \partial \log p_t(T)/ \partial T$.

Returning to (\ref{eq:VasYield}), we see that $a$ determines how strongly the short rate influences yields at longer maturities and thereby the curve steepness. Large values of $a$ imply fast convergence to the asymptotic value and hence steep yield curves, while low values of $a$ imply slow convergence and hence flat yield curves.  In Section~\ref{sec:bonds} we return to the dynamics of bond prices, and we also extend the model to include inflation bonds.

\subsection{Comments}
Except for two minor differences the model above is the same as the one used in \cite{munetal04}. We allow for a linear market price of interest rate risk, which is slightly more general than the constant market price used by \cite{munetal04}. This gives us one extra parameter for modelling yield curves. On the other hand, we limit the dependency structure slightly by assuming independence between unexpected inflation shocks ($dW^I$) and the remaining driving Brownian motions. We believe this assumption is justifiable both on theoretical grounds and from the empirical estimates provided by \cite{munetal04}, where the correlations between unexpected inflation and the other three sources of randomness are all estimated to be very close to zero.

\section{Integral representations}  \label{sec:intrep}
The primary aim of this paper is to provide the distributional results necessary for exact simulation from the model and for analytical investigation of the model's properties. The paper is also intended to serve as documentation for the derivation of these results. The results in the paper are most easily established from the following integral representations of the state variables.

For $t\geq 0$ we have
\begin{align}
  r_t & = \bar{r} + e^{-\kappa t}(r_0 - \bar{r}) + \sigma_r \int_0^t e^{-\kappa(t-s)}dW^r_s, \label{eq:rsol} \\[2mm]
  S_t & = S_0 \exp\left\{\int_0^t r_s ds + \int_0^t x_s ds - \frac{\sigma_S^2}{2}t + \sigma_S W^S_t \right\}, \label{eq:Ssol} \\[2mm]
  x_t & = \bar{x} + e^{-\alpha t}(x_0 - \bar{x}) - \sigma_x \int_0^t e^{-\alpha(t-s)}dW^S_s, \label{eq:xsol} \\[2mm]
  I_t & = I_0 \exp\left\{\int_0^t \pi_s ds - \frac{\sigma_I^2}{2}t  + \sigma_I W^I_t  \right\}, \label{eq:Isol} \\[2mm]
  \pi_t & = \bar{\pi} + e^{-\beta t}(\pi_0 - \bar{\pi}) + \sigma_\pi \int_0^t e^{-\beta(t-s)}dW^\pi_s,  \label{eq:pisol}
\end{align}
where we assume that $W^S_0 \equiv W^I_0 \equiv 0$.

The integral representations can easily be verified by It\^{o}'s lemma. For instance, to verify (\ref{eq:rsol}) write $r_t = F(t,Y_t)$, with $F(t,y)= \bar{r}+e^{-\kappa t}[(r_0-\bar{r}) + y]$ and $Y_t = \sigma_r \int_0^t e^{\kappa s}dW^r_s$. Using that $e^{-\kappa t}dY_t =
e^{-\kappa t} \sigma_r e^{\kappa t}dW^r_t = \sigma_r dW^r_t$, it follows by It\^{o}'s lemma that
\begin{align*}
   dr_t & = \frac{\partial F}{\partial t}dt + \frac{\partial F}{\partial y}dY_t + \frac{1}{2}\frac{\partial^2 F}{\partial y^2}d\langle Y, Y \rangle _t  \\[2mm]
        & = -\kappa(r_t - \bar{r})dt + e^{-\kappa t}dY_t + 0   \\[2mm]
        & = \kappa(\bar{r} - r_t)dt + \sigma_r dW^r_t.
\end{align*}
Since the right-hand side of (\ref{eq:rsol}) evaluates to $r_0$ for $t=0$ and since the differentials match we conclude, assuming sufficient regularity, that (\ref{eq:rsol}) is the solution to (\ref{eq:shortrate}).

To verify the expression for the stock index write (\ref{eq:Ssol}) as $S_t = G(t,U_t,V_t,W^S_t)$, where
$U_t = \int_0^t r_sds$, $V_t = \int_0^t x_s ds$, and
\begin{align*}
  G(t,u,v,w)=S_0\exp\left\{u+v-\frac{\sigma_S^2}{2}t + \sigma_S(w-W^S_0)\right\}.
\end{align*}
Using that $dU_t = r_t dt$, $dV_t = x_t dt$ and $d\langle W^S,W^S \rangle_t = dt$ it follows by It\^{o}'s lemma that
\begin{align*}
   dS_t & = \frac{\partial G}{\partial t}dt + \frac{\partial G}{\partial u}dU_t + \frac{\partial G}{\partial v}dV_t + \frac{\partial G}{\partial w}dW^S_t + \frac{1}{2}\frac{\partial^2 G}{\partial w^2}d\langle W^S, W^S \rangle _t  \\[2mm]
        & = -\frac{\sigma_S^2}{2}S_t dt + S_t r_t dt + S_t x_t dt + \sigma_S S_t dW^S_t + \frac{\sigma_S^2 }{2} S_t dt \\[2mm]
        & = S_t \left[ (r_t+x_t)dt + \sigma_S dW^S_t \right],
\end{align*}
which is the same as (\ref{eq:stockindex}). Since the right-hand side of (\ref{eq:Ssol}) equals $S_0$ for $t=0$ we conclude, assuming sufficient regularity, that (\ref{eq:Ssol}) is the solution to (\ref{eq:stockindex}). The three remaining expressions can be verified similarly.

\subsection{Integrated state variables}
In order to calculate the covariances of the joint distribution it is convenient to first obtain alternative expressions for the integrated short rate, excess return and expected inflation. By Fubini's theorem for stochastic integrals we can write
\begin{align}
   \int_0^t r_s ds
   & =  t \bar{r} + (r_0-\bar{r})\int_0^t e^{-\kappa s}ds + \sigma_r \int_0^t \left(\int_0^s e^{-\kappa(s-u)}dW^r_u \right)ds \nonumber \\[2mm]
   & =  t \bar{r} + (r_0-\bar{r})\int_0^t e^{-\kappa s}ds + \sigma_r \int_0^t \left(\int_u^t e^{-\kappa(s-u)}ds \right)dW^r_u \nonumber \\[2mm]
   & =  t \bar{r} + (r_0-\bar{r})\int_0^t e^{-\kappa s}ds + \sigma_r \int_0^t \left(\int_0^{t-s} e^{-\kappa u}du \right)dW^r_s \nonumber \\[2mm]
   & =  t \bar{r} + (r_0-\bar{r})\Psi(\kappa,t) + \sigma_r \int_0^t \Psi(\kappa,t-s) dW^r_s, \label{eq:intrsol}
\end{align}
where
\begin{align} \label{eq:psidef}
  \Psi(\kappa,t) \equiv \int_0^t e^{-\kappa u}du =
  \begin{cases}
  t & \mbox{for } \kappa = 0, \\
  \frac{1}{\kappa}\left(1-e^{-\kappa t} \right) & \mbox{for } \kappa \neq 0. \\
  \end{cases}
\end{align}

Similarly, we have
\begin{align}
   \int_0^t x_s ds  & =  t \bar{x} + (x_0-\bar{x})\Psi(\alpha,t) - \sigma_x \int_0^t \Psi(\alpha,t-s) dW^S_s,  \label{eq:intxsol} \\[2mm]
   \int_0^t \pi_s ds  & =  t \bar{\pi} + (\pi_0-\bar{\pi})\Psi(\beta,t) + \sigma_\pi \int_0^t \Psi(\beta,t-s) dW^\pi_s. \label{eq:intpisol}
\end{align}

Now, it follows from (\ref{eq:xsol}) that for all values of $\alpha$ and $\sigma_x$, including one or both of them zero,
\begin{align*}
    x_t - \E[x_t] & = - \sigma_x \int_0^t e^{-\alpha(t-s)}dW^S_s \\[2mm]
                  & = - \sigma_x W^S_t + \sigma_x \int_0^t 1 - e^{-\alpha(t-s)}dW^S_s \\[2mm]
                  & = - \sigma_x W^S_t + \sigma_x \alpha \int_0^t \Psi(\alpha,t-s) dW^S_s.
\end{align*}
In combination with (\ref{eq:intxsol}) this yields the relation
\begin{align}  \label{eq:linrel}
  \sigma_x W^S_t =  -\left(x_t - \E[x_t]\right) - \alpha\left(\int_0^t x_s ds - \E\left[\int_0^t x_s ds \right] \right).
\end{align}
Relation (\ref{eq:linrel}) implies that the variance-covariance matrix of the triplet $(x_t,\int_0^t x_sds,W^S_t)$ has rank at most two. If $\sigma_x \neq 0$ then $W^S_t$ is given as a linear combination of $x_t$ and $\int_0^t x_s ds$ and the rank is two, while if $\sigma_x = 0$ then $x_t$ and $\int_0^t x_s ds$ are both deterministic and the rank is one.

\section{Exact simulation} \label{sec:exactsim}
Given initial conditions $(r_0,S_0,x_0,I_0,\pi_0)$ and a sequence of increasing time points $0=t_0 < t_1 < ... < t_n$ the aim is to obtain samples from the state process observed at the discrete set of time points, i.e.\ to obtain samples of the form $\{ (r_{t_i},S_{t_i},x_{t_i},I_{t_i},\pi_{t_i}) : i=1,\ldots,n\}$. Due to the strong Markov property, this can be achieved if we can simulate successively over each time period from the joint distribution of $X_{t_i}$ given $X_{t_{i-1}}$, where $X_t=(r_t,S_t,x_t,I_t,\pi_t)'$. Due to time-homogeneity, this in turn can be achieved if we can simulate from $X_t$ given $X_0=x_{t_{i-1}}$ with $t=t_i - t_{i-1}$. By the integral representations (\ref{eq:Ssol}) and (\ref{eq:Isol}) of the stock and inflation indices, respectively, we see that samples from $X_t$ can be obtained from samples of
\begin{align}
   Y_t = \left(r_t, \int_0^t r_s ds, x_t, \int_0^t x_sds, \pi_t, \int_0^t\pi_sds, W^S_t\right)',  \label{eq:Ydef}
\end{align}
and independent draws from $W^I_t \stackrel{\mathcal{D}}{=} \sqrt{t}U$, where $U \sim N(0,1)$. The task then is to provide the (conditional) distribution of $Y_t$ given $X_0=x_0$.
This is aided by the fact that $Y_t$ follows a multivariate normal distribution, and hence we need only compute the mean vector and the variance-covariance matrix of $Y_t$.
With $Y$ being seven-dimensional this amounts to the computation of 35 quantities, i.e.\ 7 mean terms, 7 variances and 21 covariances.

\subsection{Conditional distribution}
The mean vector is readily available from the integral representations of the components of $Y_t$, while a detailed description of how to compute the variance-covariance matrix can be found in Appendix \ref{app:computation}.
For ease of reference we collect the results in the following proposition.

\begin{proposition}  \label{prop:conddist}
Let the dynamics of $r_t$, $S_t$, $x_t$, $I_t$ and $\pi_t$ be given by equations (\ref{eq:shortrate})--(\ref{eq:expectedinfl}) of Section~\ref{sec:model}, and assume that $(W^r_t,W^S_t,W^\pi_t)'$ is a vector Brownian motion with unit variance and correlation matrix
\begin{align}
   \rho =
   \begin{pmatrix}
     1  &  \rho_{rS} &  \rho_{r\pi} \\
     \rho_{rS} & 1 & \rho_{S\pi} \\
   \rho_{r\pi} & \rho_{S\pi} & 1
   \end{pmatrix}.
     \label{eq:rhomat}
\end{align}
Let $Y_t$ be given by (\ref{eq:Ydef}) for $t\geq 0$. The (conditional) distribution of $Y_t$ given initial conditions $(r_0,S_0,x_0,I_0,\pi_0)$ is multivariate normal,
\begin{align}
   Y_t | (r_0,S_0,x_0,I_0,\pi_0) \sim N_7\left( m(t,r_0,S_0,x_0,I_0,\pi_0), \Sigma(t) \right),
\end{align}
with mean vector, $m$, and variance-covariance matrix, $\Sigma$, given by Table~\ref{tab:moments} on page~\pageref{tab:moments}.
\end{proposition}
\begin{proof}
The expressions for the mean follows directly from equations (\ref{eq:rsol}), (\ref{eq:xsol}), (\ref{eq:pisol}), (\ref{eq:intrsol}), (\ref{eq:intxsol}), and (\ref{eq:intpisol}).

The variances and covariances all take the form of integrals of combinations of $e^{-\kappa_1 t}$ and $\Psi(\kappa_2,t)$. There are five cases to consider: $e^{-\kappa t}$, $\Psi(\kappa,t)$, $\Psi^2(\kappa,t)$, $e^{-\kappa_1 t}\Psi(\kappa_2,t)$, and $\Psi(\kappa_1,t)\Psi(\kappa_2,t)$, yielding expressions in terms of $\Psi$, $\Theta$, $\Upsilon$, $\Gamma$, and $\Lambda$, respectively. Definitions and computations can be found in Appendix \ref{app:computation}.
\end{proof}

\noindent
{\bf Important:} Due to relation (\ref{eq:linrel}), $\Sigma(t)$ is never of full rank. The rank is at most six, and even lower in special cases with one or more of the $\sigma$'s being zero.
The distribution of $Y_t$ is therefore singular when viewed as a 7-dimensional distribution, and care must be taken when simulating from it. Alternatively, one could consider subsets of $Y_t$ with full dimensional distributions, but this leads to a number of special cases. For ease of exposition we prefer the present unified approach.

\bigskip\noindent
Note that the mean vector depends on both $t$ and the initial conditions, while the variance-covariance matrix depends on $t$ only. The latter fact facilitates efficient simulation of the processes on an equidistant grid, say, of size $\delta$. The variance-covariance matrix, $\Sigma(\delta)$, can be computed once and the simulation can proceed stepwise, where in each step it is necessary to update only the mean vector.

Specifically, if $t_i = \delta i$ we have for $i\geq 1$
\begin{align}  \label{eq:Ysim}
   Y_{t_i} = m(\delta,r_{t_{i-1}},S_{t_{i-1}},x_{t_{i-1}},I_{t_{i-1}},\pi_{t_{i-1}}) + \Sigma(\delta)^{1/2}V,
\end{align}
where $V$ is a vector of $7$ i.i.d.\ $N(0,1)$-variates, and $\Sigma(\delta)^{1/2}$ is the Cholesky-decomposition of $\Sigma(\delta)$.\footnote{In fact, since $\Sigma(\delta)$ is not of full rank we can simulate $Y_{t_i}$ using at most 6 i.i.d.\ $N(0,1)$-variates, as shown in the following.} Having simulated $Y_{t_i}$ we then compute
\begin{align}
   S_{t_i} & = S_{t_{i-1}}\exp\left\{ Y_{t_i}^{(2)} +  Y_{t_i}^{(4)} - \frac{\sigma_S^2}{2}\delta + \sigma_S Y_{t_i}^{(7)} \right\},  \\[2mm]
   I_{t_i} & = I_{t_{i-1}} \exp\left\{ Y_{t_i}^{(6)} - \frac{\sigma_I^2}{2}\delta  + \sigma_I \sqrt{\delta}U  \right\},
\end{align}
where $U$ denotes an independent draw from $N(0,1)$. This gives us the value of all five factors at time $t_i$, $(r_{t_i}, S_{t_i}, x_{t_i}, I_{t_i}, \pi_{t_i})$, which will then serve as the conditioning variables for the simulation at time $t_{i+1}$.

In practice, the sketched algorithm might be hard to implement because $\Sigma(\delta)$ is not of full rank, and $\Sigma(\delta)^{1/2}$ is therefore difficult to compute, even though it exists theoretically. In the typical case, where the correlation matrix (\ref{eq:rhomat}) is of full rank, we can proceed as follows. First note that $\Sigma(\delta) = D C(\delta) D$, where $D = \mbox{diag}(\sigma_r,\sigma_r,\sigma_x,\sigma_x,\sigma_\pi,\sigma_\pi,1)$ and $C(\delta)$ equals $\Sigma(\delta)$ with $\sigma_r=\sigma_x=\sigma_\pi=1$. By (\ref{eq:Ysim}) we have
\begin{align}
   Y_{t_i} = m(\delta,r_{t_{i-1}},S_{t_{i-1}},x_{t_{i-1}},I_{t_{i-1}},\pi_{t_{i-1}}) + DZ,
\end{align}
where $Z \sim N_7(0,C(\delta))$, and the task then is to simulate $Z$. By assumption the variance-covariance matrix for the first six components of $Z$, $(Z_1,Z_2,Z_3,Z_4,Z_5,Z_6)$, has full rank\footnote{The claim appears intuitively reasonable, but we will not provide a proof. When the driving Brownian motions are independent, one needs to show $\Gamma^2(\kappa,\kappa,t)<\Upsilon(\kappa,t)\Psi(2\kappa,t)$ for $t>0$ and all $\kappa$.} and we can therefore simulate from this distribution by standard Cholesky-decomposition of $C(\delta)$ with its last row and column removed. Having simulated the first six components we then set $Z_7 = -Z_3 - \alpha Z_4$. It follows by (\ref{eq:linrel}) that $Z$ has the desired distribution. This procedure works also in the case where one or more of the volatilities are zero.

Simulating $Z$ when (\ref{eq:rhomat}) is not of full rank is more challenging. Perhaps the most promising approach is to use (\ref{eq:linrel}), and similar calculations for the other processes, to identify a full dimensional subset of $Z$. For example, if $W^r = c_S W^S + c_\pi W^\pi$ then $Z_1 + \kappa Z_2= c_S(Z_3+\alpha Z_4) + c_\pi(Z_5+\beta Z_6)$ and we can therefore omit $Z_1$. However, we also need to consider special cases such as $W^r=W^S$ and $\kappa=\alpha$ which imply a rank reduction of two, since $Z_1=Z_3$ and $Z_2=Z_4$. We will not attempt to give an algorithm handling all cases.

\begin{sidewaystable}
\centering
\rowcolors{2}{gray!25}{white}
\bgroup
\def\arraystretch{2}
\begin{footnotesize}
\begin{tabular}{l|ccccccc} \hline
    \rowcolor{gray!50}
  $\Cov$  & $r_t$ & $\int_0^t r_sds$ & $x_t$ & $\int_0^t x_sds$ & $\pi_t$ & $\int_0^t \pi_sds$ & $W^S_t$ \\ \hline
  $r_t$               & $\sigma_r^2\Psi(2\kappa,t)$ & & & & & & \\
  $\int_0^t r_s ds$   & $\sigma_r^2\Gamma(\kappa,\kappa,t)$ & $\sigma_r^2\Upsilon(\kappa,t)$ & & & & &  \\
  $x_t$               & $-\sigma_r\sigma_x\rho_{rS}\Psi(\kappa+\alpha,t)$ & $-\sigma_r\sigma_x\rho_{rS}\Gamma(\kappa,\alpha,t)$ & $\sigma_x^2\Psi(2\alpha,t)$ & & & &  \\
  $\int_0^t x_sds$    & $-\sigma_r\sigma_x\rho_{rS}\Gamma(\alpha,\kappa,t)$ & $-\sigma_r\sigma_x\rho_{rS}\Lambda(\kappa,\alpha,t)$ & $\sigma_x^2\Gamma(\alpha,\alpha,t)$ & $\sigma_x^2\Upsilon(\alpha,t)$ & & &  \\
  $\pi_t$             & $\sigma_r\sigma_\pi\rho_{r\pi}\Psi(\kappa+\beta,t)$ & $\sigma_r\sigma_\pi\rho_{r\pi}\Gamma(\kappa,\beta,t)$ & $-\sigma_x\sigma_\pi\rho_{S\pi}\Psi(\alpha+\beta,t)$ &
     $-\sigma_x\sigma_\pi\rho_{S\pi}\Gamma(\alpha,\beta,t)$ & $\sigma_\pi^2\Psi(2\beta,t)$ & &  \\
  $\int_0^t\pi_sds$   & $\sigma_r\sigma_\pi\rho_{r\pi}\Gamma(\beta,\kappa,t)$ & $\sigma_r\sigma_\pi\rho_{r\pi}\Lambda(\kappa,\beta,t)$ & $-\sigma_x\sigma_\pi\rho_{S\pi}\Gamma(\beta,\alpha,t)$ &
  $-\sigma_x\sigma_\pi\rho_{S\pi}\Lambda(\alpha,\beta,t)$  & $\sigma_\pi^2\Gamma(\beta,\beta,t)$ & $\sigma_\pi^2\Upsilon(\beta,t)$ &  \\
  $W^S_t$             & $\sigma_r\rho_{rS}\Psi(\kappa,t)$ & $\sigma_r\rho_{rS}\Theta(\kappa,t)$ & $-\sigma_x\Psi(\alpha,t)$ & $-\sigma_x\Theta(\alpha,t)$ & $\sigma_\pi\rho_{S\pi}\Psi(\beta,t)$ &
  $\sigma_\pi\rho_{S\pi}\Theta(\beta,t)$ & $t$ \\ \hline
      \rowcolor{gray!50}
  $\E$               & $\bar{r} + e^{-\kappa t}(r_0 - \bar{r})$ & $t\bar{r}  + \Psi(\kappa,t)(r_0 - \bar{r})$ & $\bar{x} + e^{-\alpha t}(x_0 - \bar{x})$ & $t\bar{x}  + \Psi(\alpha,t)(x_0 - \bar{x})$ &
                        $\bar{\pi} + e^{-\beta t}(\pi_0 - \bar{\pi})$ & $t\bar{\pi}  + \Psi(\beta,t)(\pi_0 - \bar{\pi})$ & 0 \\ \hline
\end{tabular}
\end{footnotesize}
\egroup
\caption{Conditional mean and variance-covariance matrix of $Y_t$ given initial conditions $X_0=(r_0,S_0,x_0,I_0,\pi_0)$, cf.\ Proposition~\ref{prop:conddist}. The matrix is symmetric, but only elements at or below the diagonal is shown to increase readability and to save space. The matrix is stated in terms of model parameters defined in equations (\ref{eq:shortrate})--(\ref{eq:expectedinfl}) of Section~\ref{sec:model}, correlation coefficients, $\rho_{rS}$, $\rho_{r\pi}$ and $\rho_{S\pi}$, between the driving Brownian motions, and the auxiliary functions $\Psi$, $\Theta$, $\Upsilon$, $\Gamma$ and $\Lambda$ defined in Appendix~\ref{app:computation}.
The variance-covariance matrix is defined for all valid model parameter values, including vanishing mean reversion parameters and singular correlation structures for the driving processes. In general, it is guaranteed that the variance-covariance matrix is positive-semidefinite, but due to (\ref{eq:linrel}) it is never positive-definite.}
\label{tab:moments}
\end{sidewaystable}

\subsection{Positive-semidefiniteness of a symmetric 3-by-3 matrix}
All parameters of the model can be freely specified, except the three correlation coefficients, which must satisfy that $\rho$ of (\ref{eq:rhomat}) constitutes a valid correlation matrix, i.e.\ that it is positive-semidefinite. This condition is easy to check for specific choices of $\rho_{rS}$, $\rho_{r\pi}$ and $\rho_{S\pi}$, e.g.\ by checking that the eigenvalues of $\rho$ are all non-negative, but for practical purposes it is preferable to have a more ``constructive'' condition.

Let us first recall the following definitions and results from linear algebra. Let $A$ be a square matrix. The determinant of a (square) submatrix of $A$ is called a {\bf minor}. It is a {\bf principal minor} if the submatrix is obtained by deleting rows and columns with the same number, and it is a {\bf leading principal minor} if the submatrix is the upper left $n$-by-$n$ corner of $A$ for some $n$. Sylvester's criterion states that a (real) symmetric matrix $A$ is positive-definite if and only if every {\it leading} principal minor is positive. The analogous necessary and sufficient condition for positive-semidefiniteness is that every principal minor of $A$ is non-negative. The latter can be turned into a useful, and easily verifiable, condition for positive-semidefiniteness of $\rho$.

\begin{proposition}  \label{prop:semidef}
Consider real numbers $x$, $y$ and $z$ in $[-1,1]$, and let
\begin{align}
   A =
   \begin{pmatrix}
     1  &  x &  y \\
     x & 1 & z \\
     y & z & 1
   \end{pmatrix}.
     \label{eq:Amat}
\end{align}
$A$ is positive-semidefinite if and only if $|A|=1-x^2 - y^2 - z^2 + 2xyz \geq 0$.
\end{proposition}
\begin{proof}
We first note that $A$ is symmetric with determinant given by
\begin{align}
  |A| = 1
   \begin{vmatrix}
     1  &  z \\
     z & 1
   \end{vmatrix}
   - x
   \begin{vmatrix}
     x & z \\
     y & 1
   \end{vmatrix}
    + y
   \begin{vmatrix}
     x & 1 \\
     y & z
   \end{vmatrix}
   = 1-x^2-y^2-z^2 + 2 xyz.
\end{align}
The ``only if''-part of the proposition follows because $|A|$ is itself a principal minor.

The ``if''-part will follow if we can show that the principal minors obtained by deleting either one or two rows and columns, are non-negative by assumption.
The principal minors obtained by deleting one row and column are
\begin{align}
   \begin{vmatrix}
     1  &  z \\
     z & 1
   \end{vmatrix}
   = 1- z^2, \quad
   \begin{vmatrix}
     1 & y \\
     y & 1
   \end{vmatrix}
     = 1- y^2, \quad
   \begin{vmatrix}
     1 & x \\
     x & 1
   \end{vmatrix}
   = 1-x^2,
\end{align}
all of which are non-negative by assumption, while deleting two rows and columns leaves the remaining diagonal element, which in all cases is $1$ and hence non-negative.
\end{proof}

To gain some intuition into the condition of Proposition~\ref{prop:semidef} we now investigate the allowed parameter space for two of the parameters, when one is fixed. Figure~\ref{fig:possemidef} shows the domain of $(x,y)$ leading to a positive-semidefinite $A$ matrix for various values of $z$. We see that for given $z$, the domain of $(x,y)$ takes the form of an ellipsis which gets more and more narrow as $z$ tends to the boundary values $\pm 1$. Interpreting $x$, $y$ and $z$ as correlation coefficients, the probabilistic content is that when two variables are highly correlated ($z$ close to $\pm 1$) then their correlations to a third variable ($x$ and $y$) must be alike. In the limiting case $z=1$ ($z=-1$), the correlation coefficients must be identical $x=y$ ($x=-y$). Note, however, that $(x,y)$ is restricted for all values of $z$, also when $z=0$.

\begin{figure}[h]
\begin{center}
\includegraphics[height=7.5cm]{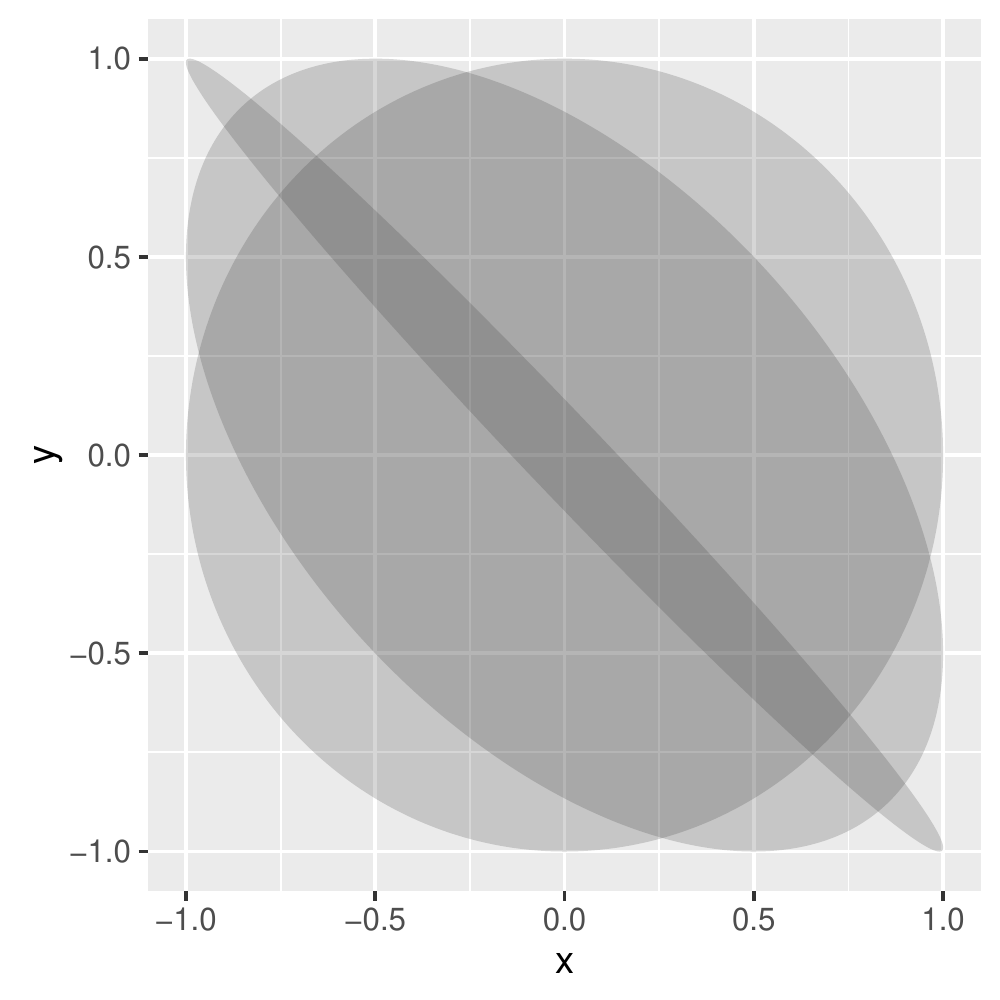}   
\hfill
\includegraphics[height=7.5cm]{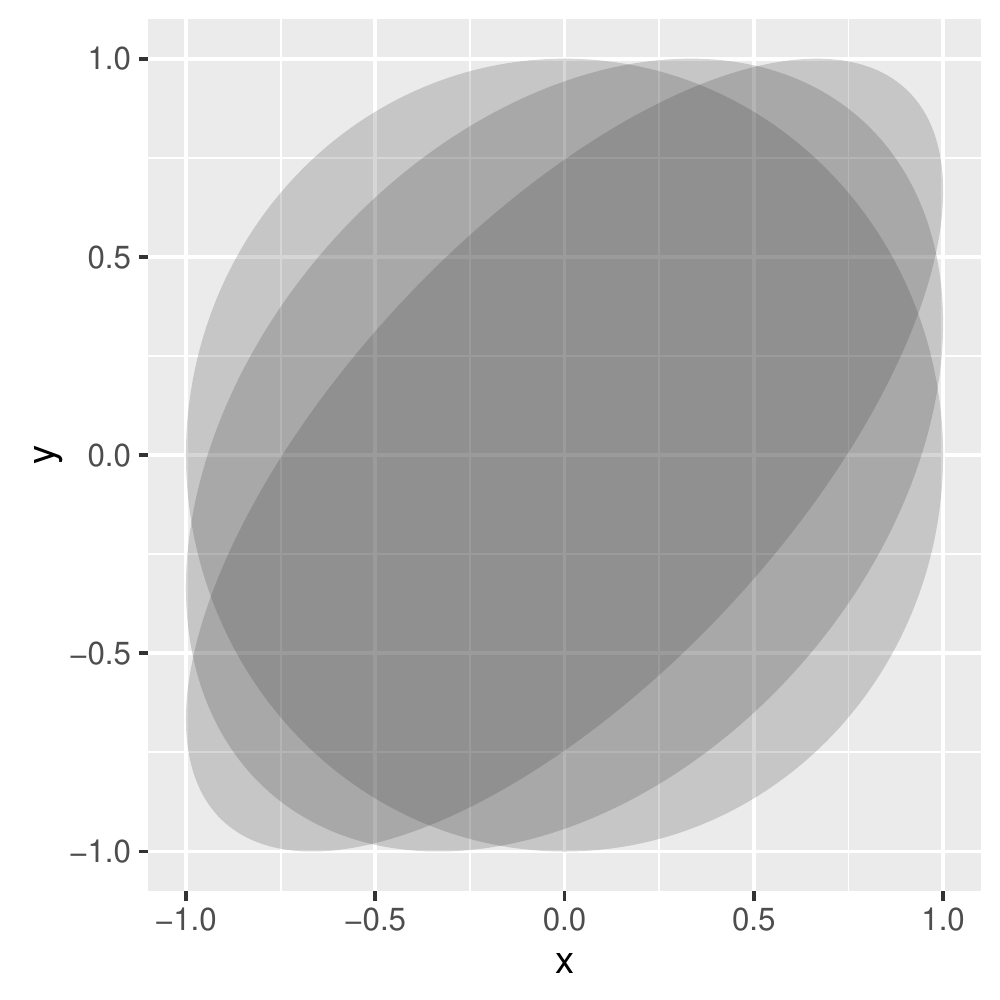}   
\end{center}
\vspace*{-5mm}
\caption{Plots of the domain of $(x,y)$ leading to a positive-semidefinite $A$ matrix for various values of $z$. In the left plot $z$ takes the values $-0.99$, $-1/2$ and $0$ (the circle), in the right plot $z$ takes the values $2/3$, $1/3$ and $0$ (the circle).}
\label{fig:possemidef}
\end{figure}

We can investigate the domains more formally by a change of variables. Consider the rotated coordinate system with orthonormal basis $\{(1/\sqrt{2},1/\sqrt{2}), (-1/\sqrt{2},1/\sqrt{2})\}$. In this basis, the coordinates of a point $(x,y)$ is given by $(u,v)=([x+y]/\sqrt{2},[x-y]/\sqrt{2})$, and conversely $(x,y)=([u+v]/\sqrt{2},[u-v]/\sqrt{2})$. Now, inserting this in the defining equation for the domain we get after a series of simple manipulations
\begin{align}
    1-x^2 - y^2 - z^2 + 2xyz \geq 0 \Leftrightarrow u^2(1-z) + v^2(1+z) \leq 1-z^2.
\end{align}
Further, for $-1 < z < 1$ we get
\begin{align}
    \frac{u^2}{1+z} + \frac{v^2}{1-z} \leq 1,
\end{align}
which shows that the domain is indeed an ellipsis with axis lengths $\sqrt{1+z}$ and $\sqrt{1-z}$, respectively. With this representation at hand we can also calculate the volume of the ``positive-semidefinite domain'' of $(x,y,z)$. For given $z$ the area of the ellipsis is $A(z) = \pi \sqrt{(1+z)(1-z)} = \pi\sqrt{1-z^2}$, and the volume is therefore
\begin{align}
    V = \int_{-1}^1 A(z)dz = \pi \int_{-1}^1 \sqrt{1-z^2}dz = \pi\left[\frac{z\sqrt{1-z^2}}{2} + \frac{1}{2}\sin^{-1}(z)\right]_{-1}^1 = \frac{\pi^2}{2}.
\end{align}
Compared to the volume of the ``bounding cube'', $[-1,1]^3$, we have that the positive-semidefinite domain occupies $\pi^2/16 \approx 61.7\%$. In probabilistic terms this implies that if the three elements (correlation coefficients) are drawn independently and uniformly at random on $[-1,1]$, then approximately 3 out of 5 times will the corresponding $A$ matrix be a valid correlation matrix.

The analysis is still primarily a method for verification of positive-semidefiniteness, but it can be recast into the following more constructive result.

\begin{proposition} \label{prop:semidefquad}
Consider real numbers $x$ and $y$ in $[-1,1]$. The matrix $A$ defined by (\ref{eq:Amat}) is a valid correlation matrix if and only if
\begin{align}
    z \in \left[ xy - \sqrt{(1-x^2)(1-y^2)}, xy + \sqrt{(1-x^2)(1-y^2)} \right].  \label{eq:psquad}
\end{align}
\end{proposition}
\begin{proof}
Consider the condition of Proposition~\ref{prop:semidef} as a condition in $z$, $f(z) = A z^2  +  B z + C \geq 0$, where $A=-1$, $B=2xy$, and $C=1-x^2-y^2$. The solution is given by the well-known formula $-B \pm \sqrt{B^2 - 4AC}/2A$, which is the same as (\ref{eq:psquad}). In principle, we should also check that (\ref{eq:psquad}) results in $z$ in $[-1,1]$, but we already know this from the preceding analysis.
\end{proof}

Due to symmetry we conclude from Proposition~\ref{prop:semidefquad} that any two correlation coefficients can be chosen freely in $[-1,1]$, while the third is then restricted to a subinterval of $[-1,1]$ of form (\ref{eq:psquad}). This concludes our analysis. Note, the analysis rests on the fact that the ``correlation'' matrix $A$ is 3-by-3, and it cannot easily be extended to higher dimensions. In the 4-by-4 case, for example, the analogue of Proposition~\ref{prop:semidef} will have five conditions, which must hold simultaneously. This complicates the analysis considerably.

\section{Nominal bonds and inflation bonds}  \label{sec:bonds}
The primary purpose of this section is to extend the model with break-even inflation (BEI) curves. Break-even inflation is the ``price'' of a real payment, i.e.\ an amount proportional to $I_T$ paid at the future time $T$. We will derive the arbitrage free BEI curve for a three-parameter specification of the market prices of inflation risk.

First, however, we revisit the formula for the term structure of nominal interest rates and derive this again using the results of the preceding sections. This serves as introduction to the risk-neutral valuation technique and it highlights the similarity in form between the term structure of nominal rates and the BEI curve. We also derive the dynamics of nominal and real bond prices and corresponding constant-maturity indices. There results clarify the role of the various market prices of risks.

\subsection{Bond prices revisited}
In Section~\ref{sec:TS} we briefly mentioned that the bond prices can de derived by use of the risk-neutral valuation formula of arbitrage theory. More precisely, the price at time $t$ of a zero-coupon bond maturing at time $T\geq t$ is given by
\begin{align}  \label{eq:Qprice}
   p_t(T) = \E_t^Q\left[\exp\left\{-\int_t^T r_s ds\right\}\right],
\end{align}
where $\E^Q_t$ denotes the conditional expectation at time $t$ taken with respect to the so-called risk neutral measure ($Q$-measure) under which the dynamics of $r$ is given by
\begin{align}  \label{eq:Qdynr}
   dr_t = a(b - r_t)dt + \sigma_r d\bar{W}^r_t.
\end{align}
The risk neutral measure in turn arises from the $P$-dynamics and a specification of the market price of (interest rate) risk, $\lambda_t^r = \lambda^r(t,r_t)$. The market price of risk measures the expected excess return per unit of interest rate risk. According to arbitrage theory $p_t(T)$ is given by (\ref{eq:Qprice}), where $Q$ is the measure under which the process defined by $d\bar{W}^r_t = dW^r_t + \lambda_t^r dt$ is a Brownian motion.  In our case we have $\lambda^r_t = \{(a-\kappa)r_t + \kappa \bar{r} - a b\}/\sigma_r$, which combined with (\ref{eq:shortrate}) yields (\ref{eq:Qdynr}). By Girsanov's theorem the $Q$-expectation can also be calculated as a $P$-expectation by use of the pricing kernel $M$,
\begin{align}
  p_t(T)= \E_t\left[\frac{M_T}{M_t} \right], \quad M_t = \exp\left\{-\int_0^t r_s ds - \frac{1}{2}\int_0^t(\lambda_s^r)^2 ds  - \int_0^t \lambda_s^r dW^r_s \right\}.
\end{align}

For illustrative purposes let us use (\ref{eq:Qprice}) to retrieve the Vasicek term structure given in Section~\ref{sec:TS}. Using that the $P$- and $Q$-dynamics of $r$ are of the same form, it follows from Proposition~\ref{prop:conddist} that
\begin{align}
   \E_t^Q\left[\int_t^T r_s ds \right]  = \Delta b + \Psi(a,\Delta)(r_t-b), \quad \Var_t^Q\left[\int_t^T r_s ds \right]  = \sigma_r^2\Upsilon(a,\Delta),
\end{align}
where $\Delta = T-t$. By the well-known formula for the mean of a log-normal variable it then follows
\begin{align}  \label{eq:Vasprice}
  p_t(T) & = \exp\left\{-\Delta b - \Psi(a,\Delta)(r_t-b) + \frac{\sigma_r^2}{2}\Upsilon(a,\Delta)\right\},
\end{align}
which is the same as the Vasicek price (\ref{eq:VasZCBPrice}), albeit in slightly different notation.

By It\^{o}'s lemma we can calculate the dynamics of $p_t(T)$ for fixed $T$. First note that $\Psi^2(a,t)= \frac{2}{a}[\Psi(a,t) - \Psi(2a,t)]$. Considering $p_t(T)$ as a function of $t$ and $r_t$ we get
\begin{align*}
  \frac{dp_t(T)}{p_t(T)} = & \frac{\partial p}{\partial t}\frac{dt}{p_t(T)} + \frac{\partial p}{\partial r}\frac{dr}{p_t(T)} + \frac{1}{2}\frac{\partial^2 p}{\partial r^2}\frac{d\langle r,r\rangle_t}{p_t(T)} \\[2mm]
       = & \left(b + e^{-a\Delta}(r_t-b) - \frac{\sigma_r^2}{2}\frac{2}{a}[\Psi(a,\Delta)-\Psi(2a,\Delta)]\right)dt \\[2mm]
         & - \Psi(a,\Delta)\left(\kappa(\bar{r}-r_t)dt + \sigma_rdW^r_t \right) + \frac{1}{2}\Psi^2(a,\Delta)\sigma_r^2dt \\
       = & (r_t - \lambda^r_t\Psi(a,T-t)\sigma_r) dt - \Psi(a,T-t)\sigma_r dW^r_t ,
\end{align*}
where we in the last line substitute $\Delta$ with $T-t$ to highlight that the second argument to $\Psi$ changes with $t$. We see that the market price of risk, $\lambda^r$, does indeed have the interpretation of measuring the expected excess return per unit of risk as claimed. Also note, that in the chosen parametrization a {\it positive} excess return corresponds to a {\it negative} $\lambda^r$, and vice versa.

Let $B_t(\tau)$ denote the price process for a {\it constant-maturity} bond index with maturity $\tau$, i.e.\ the accumulated returns from a strategy which continuously rebalances to a bond with time to maturity $\tau$. It can be shown that the dynamics of $B_t(\tau)$ is given by
\begin{align}
  \frac{dB_t(\tau)}{B_t(\tau)} = (r_t - \lambda^r_t\Psi(a,\tau)\sigma_r) dt - \Psi(a,\tau)\sigma_r dW^r_t,  \label{eq:CMNI}  
\end{align}
where, in contrast to above, the second argument to $\Psi$ is kept fixed at $\tau$. This can be shown by first deriving $d\log B_t(\tau)$ by use of (\ref{eq:Vasprice}) and a limit argument, and then using It\^{o}'s lemma, see e.g.\ Appendix A.1 of \cite{marmil14} for details of the argument. Not surprisingly, the constant-maturity bond index offers the same relation between excess return and risk, $\lambda_t^r$, as the individual bonds.

\subsection{Inflation bond prices}
We will now derive the price of an inflation bond, i.e.\ a bond which pays the value of the price index $I_T$ at the future time $T$. Since $I$ is influenced by two sources of randomness we need to specify two risk premia, $\lambda^I$ and $\lambda^\pi$. To preserve the structure of the model and retain analytic tractability we consider only risk premia of the form $\lambda_t^I=h/\sigma_I$ and $\lambda_t^\pi=\{(k-\beta)\pi_t + \beta\bar{\pi} - kl\}/\sigma_\pi$ for constants $h$, $k$ and $l$.

By arbitrage theory the price at time $t$ of an inflation bond paying $I_T$ at time $T\geq t$ is given by
\begin{align}
   q_t(T) = \E_t^Q\left[\exp\left\{-\int_t^T r_s ds\right\}I_T \right],
\end{align}
where $Q$ is the measure under which $(\bar{W}^r,\bar{W}^I,\bar{W}^\pi)$ defined by $d\bar{W}^r_t = dW^r_t + \lambda_t^r dt$, $d\bar{W}^I_t = dW^I_t + \lambda^I_t dt$ and $d\bar{W}^\pi_t = dW^\pi_t + \lambda^\pi_t dt$ is a three-dimensional Brownian motion with the same correlation structure as $(W^r,W^I,W^\pi)$ under $P$. It follows from the defining equations (\ref{eq:inflindex}) and (\ref{eq:expectedinfl}) that
\begin{align}
    \frac{dI_t}{I_t} & = (\pi_t - h)dt + \sigma_I d\bar{W}^I_t, \\[2mm]
    d\pi_t   & = k(l-\pi_t)dt + \sigma_\pi d\bar{W}^\pi_t.
\end{align}
Further, with $\Delta=T-t$ calculations similar to those in Section~\ref{sec:intrep} yield
\begin{align}
    I_T = I_t \exp\left\{\int_{t}^T \pi_s ds - h\Delta - \frac{\sigma_I^2}{2}\Delta + \sigma_I(\bar{W^I_T}-\bar{W^I_t})\right\},
\end{align}
and we can then write $q_t(T)   = I_t \E_t^Q\left[\exp(X_{t,T}) \right]$, where
\begin{align}
  X_{t,T}  =  \int_{t}^T \pi_s ds -\int_t^T r_s ds - h\Delta - \frac{\sigma_I^2}{2}\Delta + \sigma_I(\bar{W^I_T}-\bar{W^I_t}).
\end{align}
By use of Proposition~\ref{prop:conddist} we have
\begin{align}
  \E^Q_t\left[X_{t,T}\right] & =  \left(l -b- h - \frac{\sigma_I^2}{2}\right)\Delta + \Psi(k,\Delta)(\pi_t - l)  - \Psi(a,\Delta)(r_t - b), \\[2mm]
  \Var^Q_t\left[X_{t,T}\right] & =  \sigma_\pi^2\Upsilon(k,\Delta) + \sigma_r^2\Upsilon(a,\Delta) - 2\sigma_r\sigma_\pi\rho_{r\pi}\Lambda(a,k,\Delta) + \sigma_I^2 \Delta,
\end{align}
and since $X_{t,T}$ is normally distributed we can then calculate $q_t(T)$ as
\begin{align}
  q_t(T) = I_t\exp\left\{\E^Q_t\left[X_{t,T}\right] + \frac{1}{2}\Var^Q_t\left[X_{t,T}\right] \right\}.
\end{align}
By use of It\^{o}'s lemma we find after lengthy and tedious calculations, that for fixed $T$
\begin{align*}
    \frac{dq_t(T)}{q_t(T)} = & \left[r_t - \lambda^r_t\Psi(a,T-t)\sigma_r + \lambda_t^\pi \Psi(k,T-t)\sigma_\pi + \lambda_t^I\sigma_I \right]dt \\[2mm]
                           & - \Psi(a,T-t)\sigma_rdW^r_t + \Psi(k,T-t)\sigma_\pi dW^\pi_t + \sigma_I dW^I_t,
\end{align*}
showing that the risk premia are ``rewards'' for carrying the different sources of uncertainty. Note that although $W^r$ and $W^\pi$ can be correlated this does not affect the expected return.

Let $D_t(\tau)$ denote the price process for a {\it constant-maturity} inflation bond index with maturity $\tau$, i.e.\ the accumulated return from a strategy which continuously rebalances to an inflation bond with time to maturity $\tau$. It can be shown that the dynamics of $D_t(\tau)$ is given by
\begin{align}
  \begin{split}
    \frac{dD_t(\tau)}{D_t(\tau)} = & \left[r_t - \lambda^r_t\Psi(a,\tau)\sigma_r + \lambda_t^\pi \Psi(k,\tau)\sigma_\pi + \lambda_t^I\sigma_I \right]dt \\[2mm]
                           & - \Psi(a,\tau)\sigma_rdW^r_t + \Psi(k,\tau)\sigma_\pi dW^\pi_t + \sigma_I dW^I_t,  \label{eq:CMII}  
  \end{split}
\end{align}
which is the same as the dynamics of $q$, except for the constant horizon. The derivation requires a limit argument similar to the argument leading to the dynamics of the constant-maturity bond index, $B_t(\tau)$.

\subsection{Break-even inflation}
A zero-coupon inflation-indexed swap (ZCIIS) is an agreement between two parties to exchange a nominal amount for an inflation indexed amount at a future date. Specifically, for a unit notional contract settled at time $t$ the parties agree to exchange at time $T\geq t$, the notional amount $\exp\{(T-t)i_t(T)\}$ for the stochastic amount $I_T/I_t$. The yield $i_t(T)$ is termed {\it break even inflation} (BEI) and it is defined by the requirement that the contract opens at zero market value, i.e.
\begin{align}
   \frac{q_t(T)}{I_t} = p_t(T)\exp\left\{(T-t)i_t(T)\right\}.
\end{align}
It follows from the preceding two sections that
\begin{align}  \label{eq:BEI}
   i_t(T) = l - h + \frac{\Psi(k,\Delta)}{\Delta}(\pi_t-l)+ \frac{\sigma_\pi^2}{2}\frac{\Upsilon(k,\Delta)}{\Delta} - \sigma_r\sigma_\pi\rho_{r\pi}\frac{\Lambda(a,k,\Delta)}{\Delta},
\end{align}
where $\Delta = T-t$. Note, letting $\Delta$ tend to $0$ yields $i_t(t) = \pi_t-h$.

Figure~\ref{fig:BEIcurves} illustrates the effect of parameter values on the BEI curve; in the left plot the initial expected rate of inflation is $\pi_0=0\%$, and in the right plot $\pi_0=2\%$. The illustrated parameter values are given in Table~\ref{tab:parvalBEI}.  Thinking of the black line (case 1) as the base case, all subsequent cases lead to more ``expensive'' models, but for different reasons. Positive correlation to nominal rates decreases BEI, so assuming independence (case 2) leads to higher BEI. Cases 3 and 4 increase the stationary level and mean reversion strength, respectively. In the left plot where $\pi_0$ is below its stationary level, both changes increase BEI. However, in the right plot where $\pi_0=2\%$ only case 3 leads to higher BEI, while case 4 gives essentially the same curve as the base case (if $\pi_0$ were above its stationary level case 4 would lead to a BEI curve below the base case). Finally, in case 5 a negative risk premium is introduced which pushes the entire curve upwards.

\begin{table}
\centering
\rowcolors{2}{white}{gray!25}
\bgroup
\def\arraystretch{1.3}
\begin{tabular}{cccccccc} \hline
  \rowcolor{gray!50}
  Case & \multicolumn{2}{c}{Rate}  &  \multicolumn{4}{c}{Inflation} &  \multicolumn{1}{c}{Correlation} \\
  \rowcolor{gray!50}
  \#  &    $a$    & $\sigma_r$ &    $h$   &    $k$   &    $l$   &   $\sigma_\pi$  &    $\rho_{r\pi}$    \\ \hline
  1   &    0.095  &     0.01   &    0     &    0.05  &    0.02  &     0.005       &       0.80          \\
  2   &    0.095  &     0.01   &    0     &    0.05  &    0.02  &     0.005       &       {\bf 0}             \\
  3   &    0.095  &     0.01   &    0     &    0.05  &    {\bf 0.03}  &     0.005       &       0.80          \\
  4   &    0.095  &     0.01   &    0     &    {\bf 0.10}  &    0.02  &     0.005       &       0.80          \\
  5   &    0.095  &     0.01   & {\bf -0.0025}  &    0.05  &    0.02  &     0.005       &       0.80          \\ \hline
\end{tabular}
\egroup
\caption{Five sets of parameter values for illustration of BEI curves. Cases 2--5 changes one quantity ({\bf bold}) from its value in case 1, all other parameters are the same.}
\label{tab:parvalBEI}
\end{table}

\begin{figure}[h]
\begin{center}
\includegraphics[height=7.5cm]{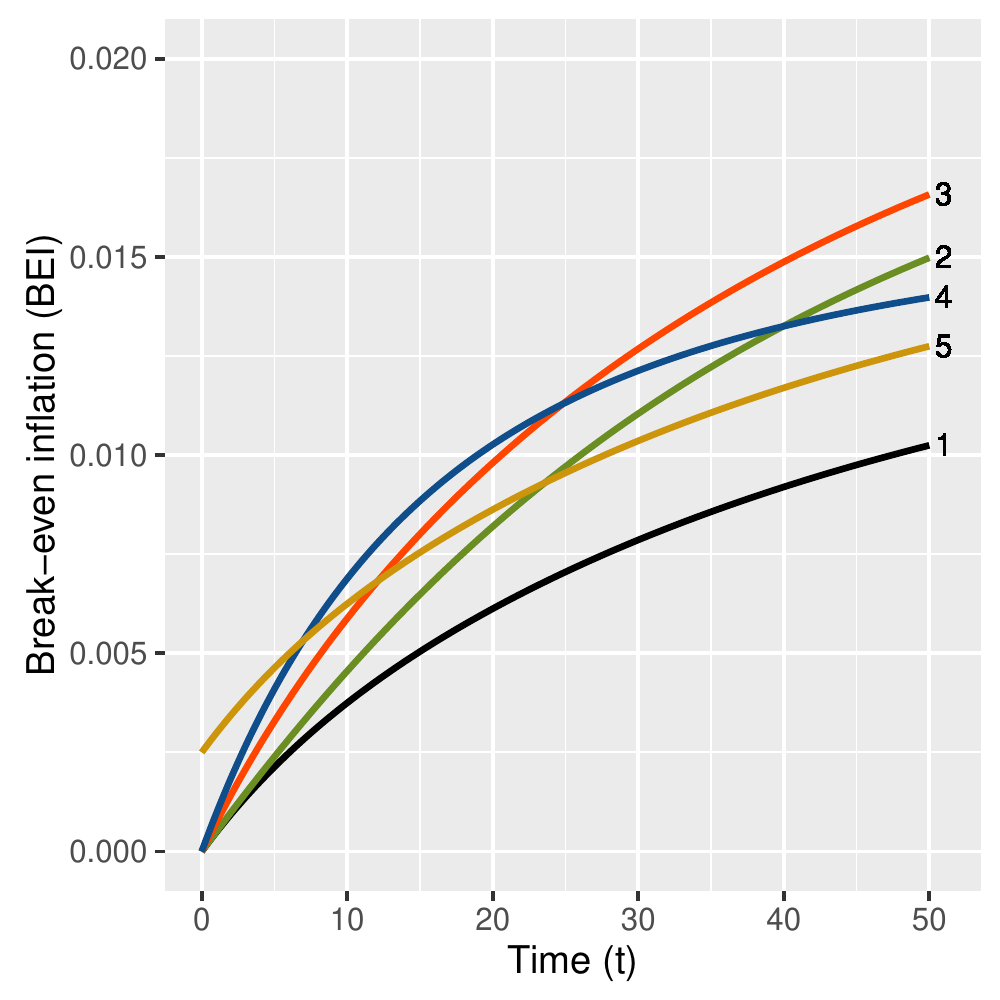}   
\hfill
\includegraphics[height=7.5cm]{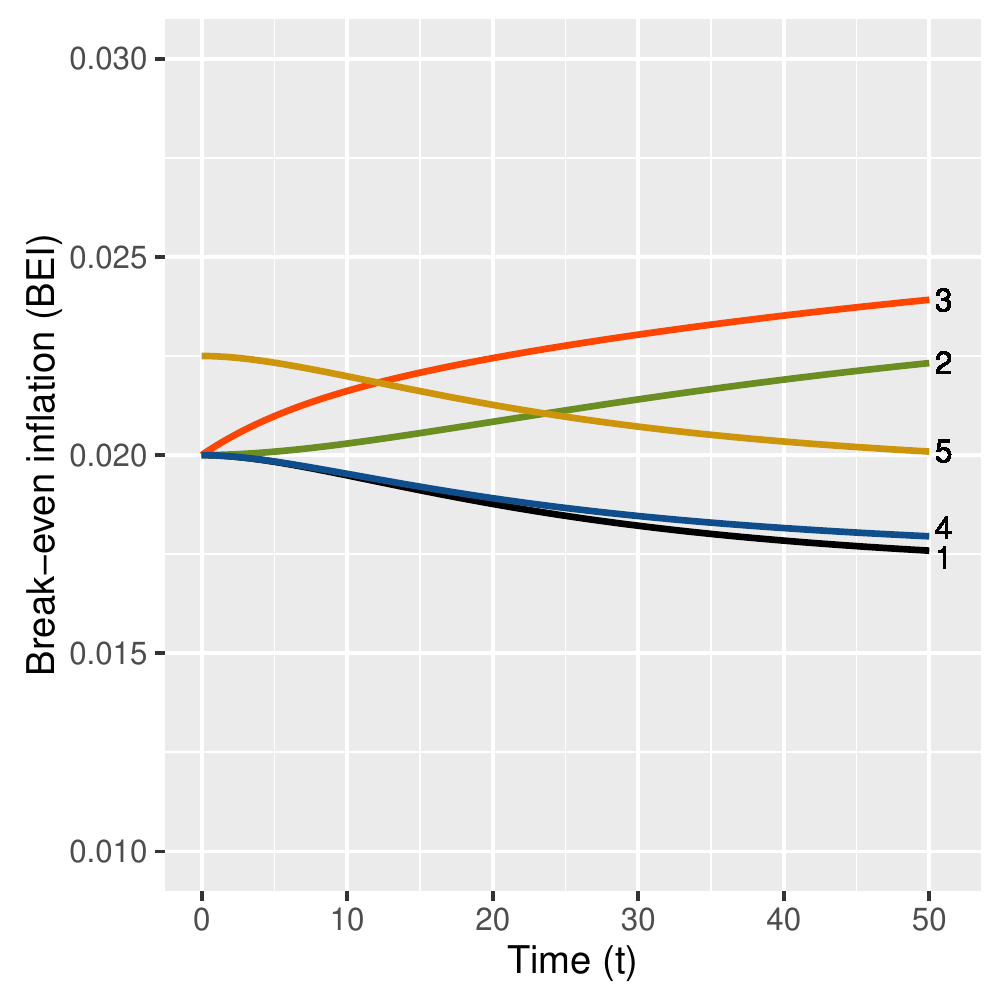}   
\end{center}
\vspace*{-5mm}
\caption{Illustration of break-even inflation (BEI) curves for the five parameter sets given in Table~\ref{tab:parvalBEI}. In the left plot $\pi_0=0\%$ and in the right plot $\pi_0=2\%$. The BEI is computed by (\ref{eq:BEI}).}
\label{fig:BEIcurves}
\end{figure}

\section{Stock and inflation indices} \label{sec:stockinfl}
The results obtained thus far allow us to derive the distribution of the stock index, $S_t$, the inflation index, $I_t$, and the stock index in real terms, $S_t/I_t$. The formulas are too complex to reveal much insight in their own right, but they can be used to investigate model properties numerically.

\begin{proposition}
For $t\geq 0$, $S_t$, $I_t$ and $S_t/I_t$ are log-normally distributed with mean
\begin{align}
   \E[\log S_t]  = & \log S_0 + \left(\bar{r}+\bar{x}- \frac{\sigma_S^2}{2} \right)t + \Psi(\kappa,t)(r_0-\bar{r}) + \Psi(\alpha,t)(x_0-\bar{x}),  \\[2mm]
   \E[\log I_t]  = & \log I_0 + \left(\bar{\pi}- \frac{\sigma_I^2}{2} \right)t + \Psi(\beta,t)(\pi_0 - \bar{\pi}),  \\[2mm]
   \begin{split}
   \E[\log S_t/I_t] = & \log \frac{S_0}{I_0} + \left(\bar{r}+\bar{x}-\bar{\pi} + \frac{\sigma_I^2}{2} - \frac{\sigma_S^2}{2} \right)t \\[2mm]
                                         & + \Psi(\kappa,t)(r_0-\bar{r}) + \Psi(\alpha,t)(x_0-\bar{x}) - \Psi(\beta,t)(\pi_0 - \bar{\pi}),
   \end{split}
\end{align}
and variances
\begin{align}
\begin{split}  \label{eq:VarlogS}
   \Var[\log S_t]  = & \sigma_r^2 \Upsilon(\kappa,t) + \sigma_x^2 \Upsilon(\alpha,t) + \sigma_S^2 t  \\[2mm]
                     & + 2\left[\sigma_r\sigma_S\rho_{rS}\Theta(\kappa,t)-\sigma_r\sigma_x\rho_{rS}\Lambda(\kappa,\alpha,t) - \sigma_x\sigma_S\Theta(\alpha,t)  \right],
\end{split}   \\[2mm]
   \Var[\log I_t]  = & \sigma_\pi^2\Upsilon(\beta,t) + \sigma_I^2 t,  \label{eq:VarlogI} \\[2mm]
\begin{split} \label{eq:VarlogSI}
   \Var[\log S_t/I_t] = & \sigma_r^2 \Upsilon(\kappa,t) + \sigma_x^2 \Upsilon(\alpha,t) + \sigma_\pi^2\Upsilon(\beta,t) + \sigma_S^2 t + \sigma_I^2 t \\[2mm]
     & + 2\left[\sigma_r\sigma_S\rho_{rS}\Theta(\kappa,t)-\sigma_r\sigma_x\rho_{rS}\Lambda(\kappa,\alpha,t) - \sigma_x\sigma_S\Theta(\alpha,t)  \right]   \\[2mm]
     & - 2\left[\sigma_S\sigma_\pi\rho_{S\pi}\Theta(\beta,t) - \sigma_x\sigma_\pi\rho_{S\pi}\Lambda(\alpha,\beta,t) + \sigma_r\sigma_\pi\rho_{r\pi}\Lambda(\kappa,\beta,t) \right],
\end{split}
\end{align}
where the functions $\Psi$, $\Theta$, $\Upsilon$ and $\Lambda$ are defined in Appendix~\ref{app:computation}.
\end{proposition}
\begin{proof}
It follows from Proposition~\ref{prop:conddist} and (\ref{eq:Ssol}) and (\ref{eq:Isol}) that $S_t$, $I_t$ and $S_t/I_t$ are log-normally distributed for $t\geq 0$. Further
\begin{align*}
  \E[\log S_t] & = \log S_0 + \E\left[ \int_0^t r_s ds + \int_0^t x_s ds - \frac{\sigma_S^2}{2}t \right], \\[2mm]
  \E[\log I_t] & = \log I_0 + \E\left[ \int_0^t \pi_s ds - \frac{\sigma_I^2}{2}t \right], \\[2mm]
  \E[\log S_t/I_t] & = \log\frac{S_0}{I_0} + \E\left[ \int_0^t r_s ds + \int_0^t x_s ds - \int_0^t \pi_s ds + \frac{\sigma_I^2}{2}t - \frac{\sigma_S^2}{2}t \right], \\[2mm]
  \Var[\log S_t] & = \Var\left[\int_0^t r_s ds + \int_0^t x_s ds + \sigma_S W^S_t \right], \\[2mm]
  \Var[\log I_t] & = \Var\left[\int_0^t \pi_s ds  + \sigma_I W^I_t \right]  \\[2mm]
  \Var[\log S_t/I_t] & = \Var\left[\int_0^t r_s ds + \int_0^t x_s ds - \int_0^t \pi_s ds  + \sigma_S W^S_t - \sigma_I W^I_t \right],
\end{align*}
from which the stated formulas follow by use of Table~\ref{tab:moments} on page~\pageref{tab:moments}, and the fact that $W^I_t$ is independent of the other quantities.
\end{proof}

Covariances between $\log S_t$, or $\log I_t$, and the other state variables can be calculated similarly, but we omit the details.

Since $\Theta$, $\Upsilon$ and $\Lambda$ are all $o(t)$ for $t$ tending to zero, it follows from (\ref{eq:VarlogS})--(\ref{eq:VarlogSI}) that
\begin{align}
  \lim_{t\to 0}\frac{\Var[\log S_t]}{t} = \sigma_S^2, \quad
  \lim_{t\to 0}\frac{\Var[\log I_t]}{t}  = \sigma_I^2, \quad
  \lim_{t\to 0}\frac{\Var[\log S_t/I_t]}{t} =  \sigma_S^2 + \sigma_I^2.  \label{eq:Varzero}
\end{align}
This of course could also be inferred directly from the defining equations (\ref{eq:stockindex}) and (\ref{eq:inflindex}), since locally only the diffusion term contributes to the variance. The limits in (\ref{eq:Varzero}) are valid for all parameter values.

In the typical situation, the Ornstein-Uhlenbeck processes for $r$, $x$ and $\pi$ are all stationary, i.e.\ the mean reversion parameters and volatilities are all positive. In this situation the leading order term of $\Theta$, $\Upsilon$ and $\Lambda$ are linear in time for large $t$ and we have the following asymptotic results
\begin{align}
  \lim_{t\to\infty}\frac{\Var[\log S_t]}{t} = & \frac{\sigma_r^2}{\kappa^2} + \frac{\sigma_x^2}{\alpha^2} + \sigma_S^2
                      + 2\left[ \frac{\sigma_r\sigma_S\rho_{rS}}{\kappa}-\frac{\sigma_r\sigma_x\rho_{rS}}{\kappa\alpha} - \frac{\sigma_x\sigma_S}{\alpha} \right], \label{eq:VarlogSlimit}\\[2mm]
  \lim_{t\to\infty}\frac{\Var[\log I_t]}{t}  = & \frac{\sigma_\pi^2}{\beta^2} + \sigma_I^2, \label{eq:VarlogIlimit} \\[2mm]
\begin{split}   \label{eq:VarlogSIlimit}
  \lim_{t\to\infty}\frac{\Var[\log S_t/I_t]}{t} = & \frac{\sigma_r^2}{\kappa^2} + \frac{\sigma_x^2}{\alpha^2} + \sigma_S^2
       + 2\left[ \frac{\sigma_r\sigma_S\rho_{rS}}{\kappa}-\frac{\sigma_r\sigma_x\rho_{rS}}{\kappa\alpha} - \frac{\sigma_x\sigma_S}{\alpha} \right] \\[2mm]
     & + \frac{\sigma_\pi^2}{\beta^2} + \sigma_I^2 - 2\left[\frac{\sigma_S\sigma_\pi\rho_{S\pi}}{\beta} - \frac{\sigma_x\sigma_\pi\rho_{S\pi}}{\alpha\beta} + \frac{\sigma_r\sigma_\pi\rho_{r\pi}}{\kappa\beta} \right].
\end{split}
\end{align}
Hence, in a sense the processes $S$, $I$ and $S/I$ all behave like geometric Brownian motions on sufficiently long horizons. Note, however, that depending on parameter values the convergence can be rather slow and a rich set of variance structures on both moderate and longer time horizons are in fact possible.

The limit results are also valid if one or more of $r$, $x$ and $\pi$ are deterministic, i.e.\ if one or more of the volatilities are zero (with the convention $0/0=0$ the formulas even hold when the volatility and the corresponding mean reversion parameter are both zero, e.g.\ $\alpha=\sigma_x=0$). The limit results do not hold, however, in the random walk or exploding cases where the ``mean reversion'' parameter is zero or negative and the volatility is positive.

\subsection{Variance structure}
In the widely used Black-Scholes model the stock index is modelled as a geometric Brownian motion (GBM)\footnote{The Black-Scholes model is a special case of our model obtained by setting $\kappa=\sigma_r=\alpha=\sigma_x=0$.},
\begin{align}
    \frac{dS_t}{S_t} = (r + \mu)dt + \sigma_S dW^S_t,
\end{align}
where both the (short) rate, $r$, and the risk premium, $\mu$, are constant. The stock index is log-normally distributed with
\begin{align}
    \E[\log S_t] & = \log S_0 + \left(r+\mu - \frac{\sigma_S^2}{2}\right)t, \quad \Var[\log S_t] = \sigma_S^2 t.
\end{align}
The model thus implies that the (short term) volatility accumulates linearly over time. This is problematic when simulating over long horizons because the distribution of the stock index becomes very wide (for realistic values of short term stock volatility $\sigma_S$). In other words, the implicitly assumed linear relationship between instantaneous and long-term variance is undesirable.

One of the key features of the model of this paper is that it offers a more flexible variance structure. This is achieved by the risk premium being negatively correlated with the stock index, which in turn implies mean reverting stock returns. We can therefore have both a realistic level of short term stock volatility and control of the width of the stock index distribution on longer horizons. As seen previously, the variance will ultimately be linearly increasing but the rate can differ from the short term volatility.

In the following we illustrate how the various parameters affect the volatility of $\log S$ and $\log S/I$ over time. We assume throughout that $\sigma_S=0.15$ corresponding to the volatility of a well-diversified
stock portfolio, and $\sigma_I=0.005$ corresponding to only small deviations between realized and expected inflation. For the other parameters we consider the cases shown in Table~\ref{tab:parval}. Note that since we are only interested in the variation we do not need to specify the mean reversion levels, nor the initial condition for the state variables.

\begin{table}
\centering
\rowcolors{2}{white}{gray!25}
\bgroup
\def\arraystretch{1.3}
\begin{tabular}{cccccccccccc} \hline
  \rowcolor{gray!50}
  Case & \multicolumn{2}{c}{Rate}  &  \multicolumn{2}{c}{Stock} &  \multicolumn{2}{c}{Inflation} &  \multicolumn{3}{c}{Correlations} &  \multicolumn{2}{c}{Asymp. vol.} \\
  \rowcolor{gray!50}
  \#  & $\kappa$  & $\sigma_r$ & $\alpha$ & $\sigma_x$ & $\beta$ & $\sigma_\pi$  & $\rho_{rS}$ & $\rho_{r\pi}$ & $\rho_{S\pi}$  & Nom.  & Real \\ \hline
  1   &    0      &     0      &   0     &      0      &   0      &    0          &    0      &      0       &          0       &  0.150 &  0.150    \\
  2   &    0.05   &     0.01   &   0     &      0      &   0.05   &    0.005      &    0      &      0.80    &         -0.25    &  0.250 &  0.219    \\
  3   &    0.05   &     0.01   &   0.06  &      0.007  &   0.05   &    0.005      &    0      &      0.80    &         -0.25    &  0.203 &  0.144    \\
  4   &    0.05   &     0.01   &   0.06  &      0.015  &   0.05   &    0.005      &    0      &      0.80    &         -0.25    &  0.224 &  0.152     \\
  5   &    0.10   &     0.01   &   0.06  &      0.015  &   0.05   &    0.005      &    0      &      0.80    &         -0.25    &  0.141 &  0.095    \\ \hline
\end{tabular}
\egroup
\caption{Five sets of parameter values used for illustration of different variance structures of $\log S$ and $\log S/I$. Mean reversion levels ($\bar{r}$, $\bar{x}$ and $\bar{\pi}$) and initial conditions of the state variables are not specified since they do not affect the variance. The last two columns show the asymptotic volatility (rate) for $\log S$ and $\log S/I$, i.e.\ the square root of the right-hand side of (\ref{eq:VarlogSlimit}) and (\ref{eq:VarlogSIlimit}), respectively.}
\label{tab:parval}
\end{table}

\begin{figure}[h]
\begin{center}
\includegraphics[height=7.5cm]{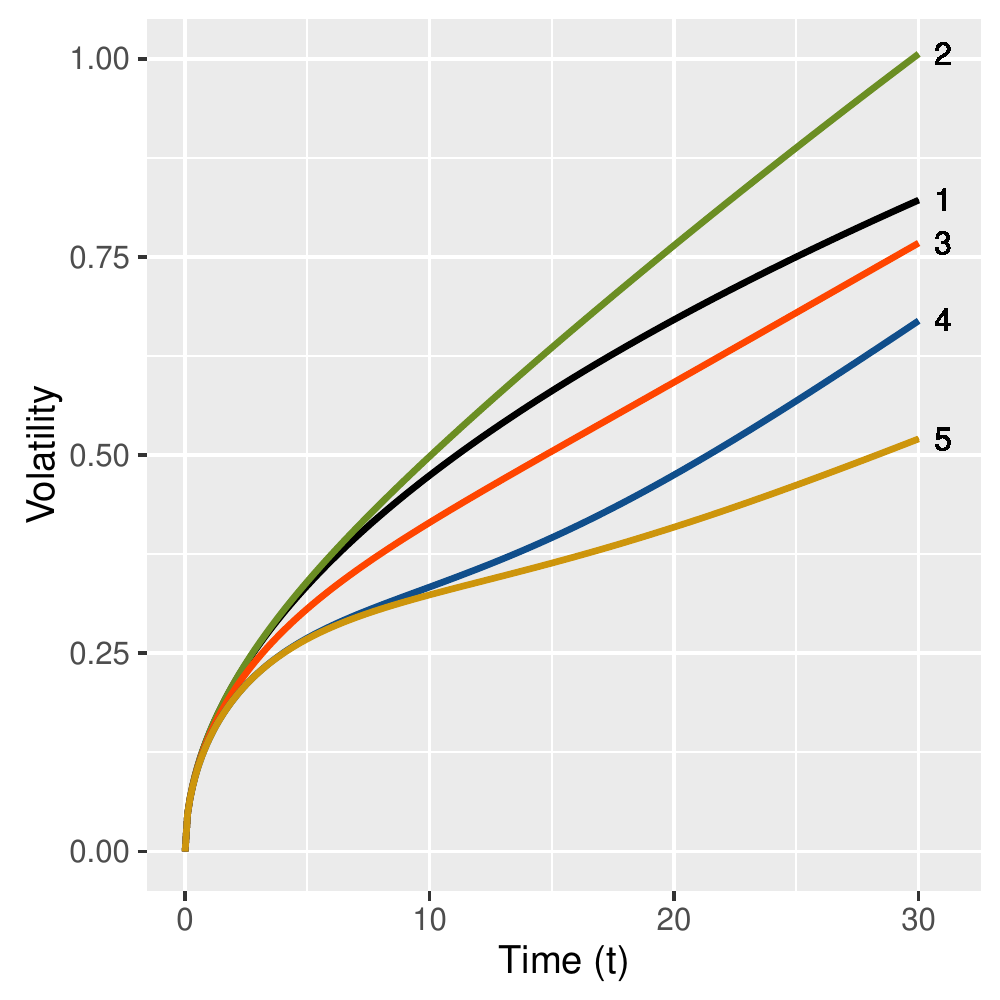}   
\hfill
\includegraphics[height=7.5cm]{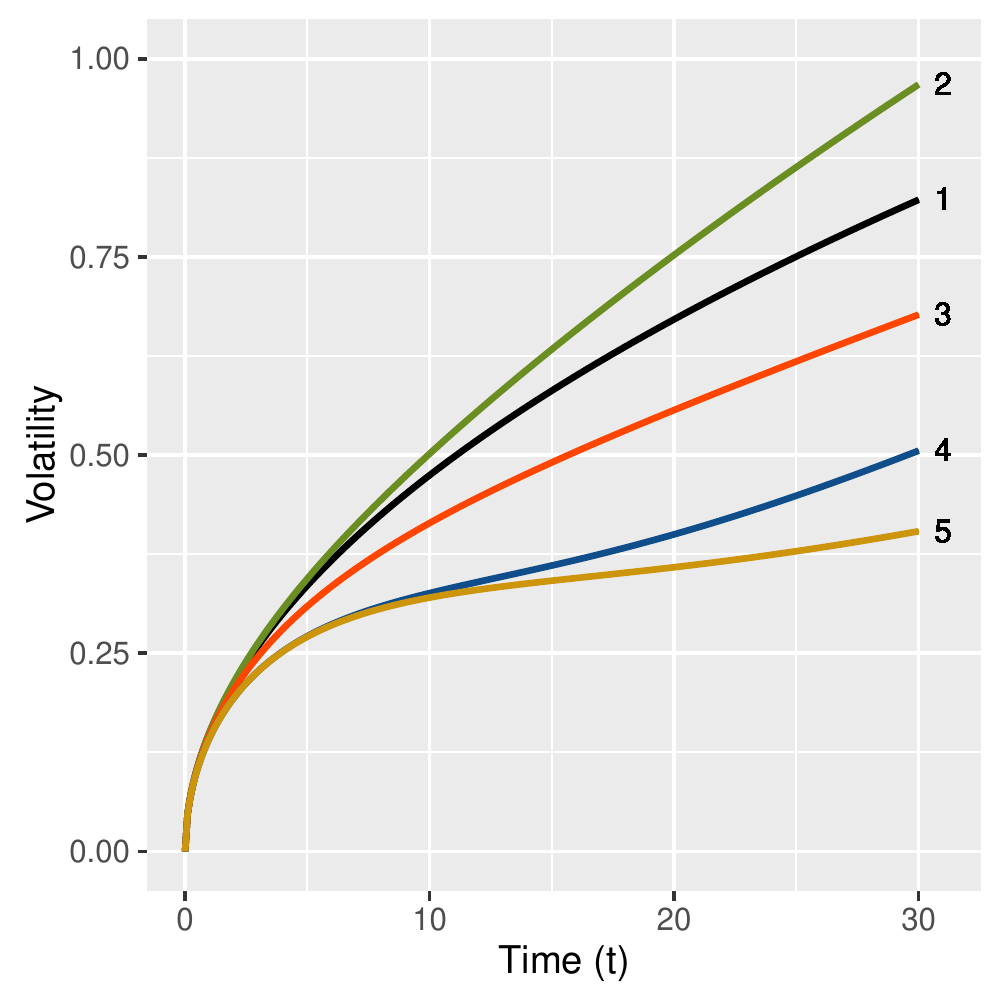}   
\end{center}
\vspace*{-5mm}
\caption{Illustration of the standard deviation (volatility) of $\log S$ (left plot) and $\log S/I$ (right plot) as a function of time. The standard deviation is computed as the square root of the right-hand side of (\ref{eq:VarlogS}) and (\ref{eq:VarlogSI}), respectively. The numbers refer to the five parameter sets given in Table~\ref{tab:parval}.}
\label{fig:volatility}
\end{figure}

Figure~\ref{fig:volatility} shows the standard deviation of $\log S$ (left plot) and $\log S/I$ (right plot) over time for the five parameter sets given in Table~\ref{tab:parval}. The black line (case 1) corresponds to a Black-Scholes model for both $S$ and $S/I$ (the volatility in the latter case is $(\sigma_S^2 + \sigma_I^2)^{1/2} = 0.1501 \approx \sigma_S$). In this case the standard deviation is proportional to $\sqrt{t}$. Adding stochastic interest rates and expected inflation (case 2) increases the volatility. However, due to the assumed positive correlation between $r$ and $\pi$ the value of the stock index in real terms ($S/I$) is less volatile than its nominal value ($S$).

Looking at Figure~\ref{fig:volatility}, we see that introducing mean reverting stock returns (case 3) decreases the volatility of $\log S$ in both nominal and real terms, increasing the feedback on expected returns (case 4) decreases the volatility further, and narrowing the interest rate distribution (case 5) decreases the volatility even further. These results conform with the rationale for including the various model components and they illustrate different ways of controlling the stock index volatility on longer horizons. Note, however, that asymptotically only case 5 has lower volatility than the Black-Scholes model, thus eventually the blue and orange lines cross the black line. This perhaps is somewhat surprising and it demonstrates the limitation of intuitive reasoning --- and the usefulness of formal analysis.

\section{Portfolio properties} \label{sec:portfolio}
The model offers the possibility of investing in stocks, nominal bonds, and inflation bonds, in addition to cash. In this section we discuss properties of portfolios constructed from these assets.

\subsection{Dynamics and Sharpe ratio}
To avoid a discussion of rebalancing rules and holding periods, we assume in the following that nominal and inflation bond exposure is taken by investing in constant-maturity indices. Specifically, we assume that we invest in the (nominal) index $B_t$ with maturity $\tau_B$, and in the (real) index $D_t$ with maturity $\tau_D$. The dynamics of $B_t(\tau_B)$ and $D_t(\tau_D)$ are given by (\ref{eq:CMNI}) and (\ref{eq:CMII}), respectively.

Let us consider a portfolio with continuous rebalancing to desired relative weights: $w^S$ in stocks, $w^B$ in nominal bonds (index $B$), $w^D$ in inflation bonds (index $D$), and the rest in cash. We place no restrictions on the weights. Initially, we have the ``classical'' case in mind, i.e.\ the weights are chosen to give a balanced asset mix (regardless of the assets' exposure to the underlying risk factors). Of course, this only influences the way we interpret the formulas; the mathematics is the same regardless of how we determine the weights. The dynamics of the portfolio value, $V_t$, is readily available from equations (\ref{eq:stockindex}), (\ref{eq:CMNI}) and (\ref{eq:CMII})
\begin{align*}
  \frac{dV_t}{V_t} & = w^S \frac{dS_t}{S_t} + w^B \frac{dB_t(\tau_B)}{B_t(\tau_B)} + w^D \frac{dD_t(\tau_D)}{D_t(\tau_D)} + \left(1-w^S-w^B-w^D\right)r_t dt \\[2mm]
           & = \left[ r_t  + w'L \lambda_t \right] dt + w'LdW_t,
\end{align*}
where $w=(w^S,w^B,w^D)'$, $\lambda_t = (\lambda^S_t, \lambda^r_t, \lambda^\pi_t, \lambda^I_t)'$ with $\lambda^S_t=x_t/\sigma_S$, $W_t = (W^S_t,W^r_t,W^\pi_t,W^I_t)'$, and\footnote{Note that we use a different ordering of the underlying Brownian motions here that in Proposition~\ref{prop:conddist}. This is purely for aesthetical reasons, the structure of $L$ is (visually) simpler when the risk factors relating to nominal and inflation bonds are grouped together.}
\begin{align}
   L =  \begin{pmatrix}
   \sigma_S   &    0                         &  0 & 0   \\
    0         &    -\Psi(a,\tau_B)\sigma_r   &  0 & 0   \\
    0         &    -\Psi(a,\tau_D)\sigma_r   &  \Psi(k,\tau_D)\sigma_\pi & \sigma_I
   \end{pmatrix}.
\end{align}

The (instantaneous) Sharpe ratio is defined as the ratio of the (annualized) expected excess return to the (annualized) volatility, i.e.\
\begin{align}
   R_t = \frac{w'L\lambda_t}{\sqrt{w'LCL'w}},  \label{eq:R}
\end{align}
where the correlation matrix $C$ is given by
\begin{align}
   C =  \begin{pmatrix}
   1           &    \rho_{rS}     & \rho_{S\pi}   & 0   \\
   \rho_{rS}   &      1           & \rho_{r\pi}   & 0   \\
   \rho_{S\pi} &    \rho_{r\pi}   &    1          & 0   \\
   0           &       0          &    0          & 1
   \end{pmatrix}.
\end{align}
The risk-reward measure $R$ is the continuous time analogue to the well-known Sharpe ratio performance measure (in discrete time), and it has a similar interpretation as a measure of effectiveness. For $w=(1,0,0)'$ and $w=(0,1,0)'$ we get $R_t = x_t/\sigma_S$ and $R_t = -\lambda^r_t$, respectively. In these cases $R$ measures the (marginal) risk-reward trade-off for one specific risk source, and $R$ is then also known as the {\it market price of risk} (for that risk source).

Note that the numerator of $R$ is time-dependent, while the denominator is not. This reflects the fact that the assets under consideration have constant volatility over time, while the risk premia are in general stochastic (except for $\lambda^I_t$) . Consequently, the Sharpe ratio is also stochastic.

\begin{proposition}  \label{prop:SR}
For $t\geq 0$, the Sharpe ratio, $R_t$, is normally distributed with
\begin{align}
   \E[R_t] = \frac{w'L\E[\lambda_t]}{\sqrt{w'LCL'w}},  \quad   \Var[R_t] = \frac{w'MV(t)M'w}{w'LCL'w},
\end{align}
where
\begin{align}
\E[\lambda_t] =
    \begin{pmatrix}
      \bar{x}/\sigma_S \\[2mm]
      a(\bar{r}-b)/\sigma_r \\[2mm]
      k(\bar{\pi}-l)/\sigma_\pi \\[2mm]
      h/\sigma_I
   \end{pmatrix}
   +
    \begin{pmatrix}
     e^{-\alpha t}(x_0-\bar{x})/\sigma_S \\[2mm]
     e^{-\kappa t}(a-\kappa)(r_0-\bar{r})/\sigma_r \\[2mm]
     e^{-\beta t}(k-\beta)(\pi_0-\bar{\pi})/\sigma_\pi \\[2mm]
     0
   \end{pmatrix},   \label{eq:meanlambda}
\end{align}
\begin{align}
   M =  \begin{pmatrix}
   \sigma_x   &    0                         &  0    \\[2mm]
    0         &    -\Psi(a,\tau_B)(a-\kappa)\sigma_r   &  0    \\[2mm]
    0         &    -\Psi(a,\tau_D)(a-\kappa)\sigma_r   &  \Psi(k,\tau_D)(k-\beta)\sigma_\pi
   \end{pmatrix},
\end{align}
and
\begin{align}
   V(t) =  \begin{pmatrix}
   \Psi(2\alpha,t)                     &  -\rho_{rS}\Psi(\kappa+\alpha,t)    &   -\rho_{S\pi}\Psi(\alpha+\beta,t)  \\[2mm]
   -\rho_{rS}\Psi(\kappa+\alpha,t)     &    \Psi(2\kappa,t)                  &    \rho_{r\pi}\Psi(\kappa+\beta,t)  \\[2mm]
   -\rho_{S\pi}\Psi(\alpha+\beta,t)    &   \rho_{r\pi}\Psi(\kappa+\beta,t)   &   \Psi(2\beta,t)
   \end{pmatrix}.
\end{align}
\end{proposition}
\begin{proof}
Recall that the risk premia are defined by $\lambda^S_t=x_t/\sigma_S$, $\lambda^r_t=\{(a-\kappa)r_t + \kappa \bar{r} - a b\}/\sigma_r$, $\lambda_t^\pi=\{(k-\beta)\pi_t + \beta\bar{\pi} - kl\}/\sigma_\pi$, and $\lambda_t^I=h/\sigma_I$. In combination with (\ref{eq:R}) it follows that $R_t$ is a linear combination of the state variables $x_t$, $r_t$ and $\pi_t$. By Proposition~\ref{prop:conddist} $(x_t,r_t,\pi_t)$ is jointly normal, and $R_t$ is therefore also normally distributed. The stated expressions for the mean and variance follows from Table~\ref{tab:moments} and the relation $L\Var[\lambda_t]L'=MV(t)M'$, which can easily be verified.
\end{proof}

In the typical (stationary) situation where $\alpha$, $\kappa$ and $\beta$ are all positive, the last vector on the right-hand side of (\ref{eq:meanlambda}) converges to zero as $t$ tends to infinity. Hence, the first vector on the right-hand side of (\ref{eq:meanlambda}) represents the long-term (stationary) risk premia. Note that the last vector has entries of zero if the state processes $x$, $r$ and $\pi$ are started at their stationary mean, or---in the case of $r$ and $\pi$---if the mean reversion parameters are the same under $P$ and $Q$ ($a=\kappa$ and $k=\beta$).

\subsection{Factor investing}
In the factor investing paradigm the focus is on exposure to (common) risk factors, rather than on exposure to specific assets classes. This idea is well illustrated in the present model where the common factors are easily  identified. For a general treatment of factor investing the interested reader is referred to \cite{ang14}.

Assume the investor wants exposure to equity risk, interest rate risk and inflation risk. Assume further that the risk distribution to each of these factors is given by the triplet $f=(f^E,f^R,f^I)$. First, we need to decide which risk sources contribute to the different factors. Due to the simplicity of the model this is rather straightforward: risk originating from $W^S$ is equity risk, $W^r$ is interest rate risk, and $W^\pi$ and $W^I$ are inflation risk. Second, we need to decide which risk measure to use and how to aggregate risk from more than one source. Since we are working with Gaussian processes where (essentially) all risk measures are equivalent the obvious choice is to use volatility as risk measure. Regarding risk aggregation, we need to combine the risk from $W^\pi$ and $W^I$ into one inflation risk measure. Again, due to risk sources being Gaussian (and independent by assumption) it seems a natural choice to use the volatility from the combined impact. Note, however, under other distributional assumptions an argument could also be made for other ways to combine risk, or indeed other risk measures, e.g.\ expected shortfall.

Considering the same assets as in the previous section we want to find the (relative) allocation that gives the desired factor exposure. With the choices made above, this is equivalent to solving $w'\tilde{L}=f$ for $w$, where
\begin{align}
    \tilde{L} =  \begin{pmatrix}
   \sigma_S   &    0                         &  0    \\
    0         &    \Psi(a,\tau_B)\sigma_r   &  0    \\
    0         &    \Psi(a,\tau_D)\sigma_r   &  \sqrt{\Psi^2(k,\tau_D)\sigma_\pi^2 + \sigma_I^2}
   \end{pmatrix}.
\end{align}
Note that the entries corresponding to interest rate risk are positive in $\tilde{L}$, but negative in $L$. This is simply to conform with the convention that a long position in a nominal bond has positive interest rate risk (or, equivalently, that the underlying risk source is $-W^r_t$).

In general, solving for $w$ requires inverting $\tilde{L}$, but in the present case the structure is so simple that the solution can be written down directly
\begin{align}
  w^S & = \frac{f^E}{\sigma_S}, \label{eq:wS} \\[2mm]
  w^B & = \frac{f^R}{\Psi(a,\tau_B)\sigma_r} - \frac{\Psi(a,\tau_D)}{\Psi(a,\tau_B)}\frac{f^I}{\sqrt{\Psi^2(k,\tau_D)\sigma_\pi^2 + \sigma_I^2}},  \label{eq:wB} \\[2mm]
  w^D & = \frac{f^I}{\sqrt{\Psi^2(k,\tau_D)\sigma_\pi^2 + \sigma_I^2}}.  \label{eq:wD}
\end{align}
The solution (\ref{eq:wS})--(\ref{eq:wD}) gives the relative cash allocation to the three assets. This allocation can then be scaled to meet a desired total risk target; if the desired risk target is given in terms of a volatility target of, say, $\sigma_{tot}$, then each of the $w$'s is multiplied by the factor  $\sigma_{tot}/\sqrt{w'LCL'w}$.

Note, the joint distribution enters only in the second stage when the total risk is scaled to the desired level. In the first stage we consider only the marginal risk contribution from each factor. Of course, the (relative) factor risk target, $f$, may itself be the result of a joint analysis of the factors.

\subsection{Example}
Let us consider a factor investor with factor risk allocation $(f^E,f^R,f^I)=(40,40,20)$. Let us further assume that interest rate and inflation risk are taken by investing in constant-maturity indices with maturities $\tau_B=10$ and $\tau_D=15$ years, respectively. We consider the capital market model with parameters given by Table~\ref{tab:cappar}. We note that the market price of interest rate risk is initially negative. It does however depend (positively) on the interest rate level, and it is therefore expected to increase over time. The market price of both expected ($\lambda^\pi$) and unexpected ($\lambda^I$) inflation risk are also negative. These market prices are constant and hence negative at all horizons.

\begin{table}
\centering
\rowcolors{2}{gray!25}{white}
\bgroup
\def\arraystretch{1.3}
\begin{tabular}{C{1.2cm}C{1.2cm}C{1.2cm}C{1.2cm}C{1.2cm}C{1.2cm}C{1.2cm}C{1.2cm}} \hline
  \rowcolor{gray!50}
  \multicolumn{2}{c}{Rate}  &  \multicolumn{2}{c}{Stock} &  \multicolumn{2}{c}{Inflation}  &  \multicolumn{2}{c}{Correlations} \\ \hline
    $\kappa$    &  0.09     &   $\alpha$   &  0.06       &   $\beta$       &   0.05        &   $\rho_{rS}$     & 0       \\
    $\bar{r}$   &  0.0275   &   $\bar{x}$  &  0.045      &   $\bar{\pi}$   &   0.015        &   $\rho_{r\pi}$   & 0.80    \\
    $\sigma_r$  &  0.01     &   $\sigma_x$ &  0.007      &   $\sigma_\pi$  &   0.005       &   $\rho_{S\pi}$   & -0.25   \\
                &           &   $\sigma_S$ &  0.15       &   $\sigma_I$    &   0.005       &                   &         \\ \hline
    $a$         &  0.03    &              &             &   $k$           &    0.05       &                   &         \\
    $b$         &  0.065    &              &             &   $l$           &    0.02       &                   &         \\
                &           &              &             &   $h$           &    -0.001    &                   &         \\ \hline
  \rowcolor{gray!50}
  \multicolumn{8}{c}{Initial values of state variables and market prices of risk}  \\  \hline
  $r_0$         &   0.005   &     $x_0$           & 0.03        &   $\pi_0$               &    0      &                    &         \\
  $-\lambda^r_0$ &  -0.0225   &     $\lambda^S_0$   & 0.20        &   $\lambda^\pi_0$       &   -0.05   &   $\lambda^I_0$    &    -0.20     \\ \hline
\end{tabular}
\egroup
\caption{Capital market parameters used for illustration of Sharpe ratios in Figure~\ref{fig:SRexample}. Top panel shows parameters governing the physical dynamics ($P$-measure), middle panel shows pricing parameters ($Q$-measure), and bottom panel shows initial values for the state variables and market prices of risk. The market prices of risk are given by $-\lambda^r_t = 6 r_t - 0.0525$, $\lambda^S_t = x_t/0.15$, $\lambda^\pi_t = -0.05$, and $\lambda^I_t=-0.2$.}
\label{tab:cappar}
\end{table}

We first calculate the corresponding cash allocation by use of (\ref{eq:wS})--(\ref{eq:wD}). Scaling the weights to sum to $100$, we find $(w^S,w^B,w^D)=(46,-11,65)$. Since inflation is rather stable it requires a large (cash) exposure to meet the desired inflation risk target. The inflation bond exposure entails a large interest rate risk also, and it is therefore necessary to short nominal bonds. The shorting is solely an effect of risk contributions, and is unrelated to the market prices of risk.

We note that initially the exposure to nominal rates and inflation both {\it reduce} the expected return of the portfolio. The nominal bond index has a volatility of $\Psi(a,\tau_B)\sigma_r=8.6\%$ and an initial excess return of $-\lambda^r_0\Psi(a,\tau_B)\sigma_r = -0.19\%$. The inflation bond index has a volatility of $8.5\%$ and an initial excess return of $-0.64\%$, yielding an initial Sharpe ratio of $-0.64\%/8.5\%=-0.075$.\footnote{By (\ref{eq:CMII}) the variance and excess return are given by $\Psi^2(a,\tau_D)\sigma_r^2+\Psi^2(k,\tau_D)\sigma_\pi^2-2\Psi(a,\tau_D)\sigma_r\Psi(k,\tau_D)\sigma_\pi\rho_{r\pi} + \sigma_I^2 = (12.1\%)^2 + (5.3\%)^2 - 2(12.1\%)(5.3\%)0.8 + (0.5\%)^2 = (8.5\%)^2$, and
$-\lambda^r_0\Psi(a,\tau_D)\sigma_r + \lambda^\pi_0\Psi(k,\tau_D)\sigma_\pi + \lambda^I_0\sigma_I = -0.272\% - 0.264\% - 0.100\% = -0.64\%$, respectively.}

The left plot of Figure~\ref{fig:SRexample} shows the expected Sharpe ratio over time for the three assets in the portfolio and for the portfolio itself. The interest rate and the expected excess return for stocks ($x$) are both started below their stationary levels and hence tend to drift upwards. The expected Sharpe ratio for all assets, and hence for the portfolio, is therefore increasing over time. The right plot of Figure~\ref{fig:SRexample} shows the Sharpe ratio distribution at time 15 for the three assets and the portfolio. As seen there is substantial variation around the expected value. Perhaps surprisingly, the portfolio Sharpe ratio is more volatile than the Sharpe ratio for both stocks and nominal bonds, but less volatile than for inflation bonds. The Sharpe ratio distribution is calculated by use of Proposition~\ref{prop:SR}.

\begin{figure}[h]
\begin{center}
\includegraphics[height=7.5cm]{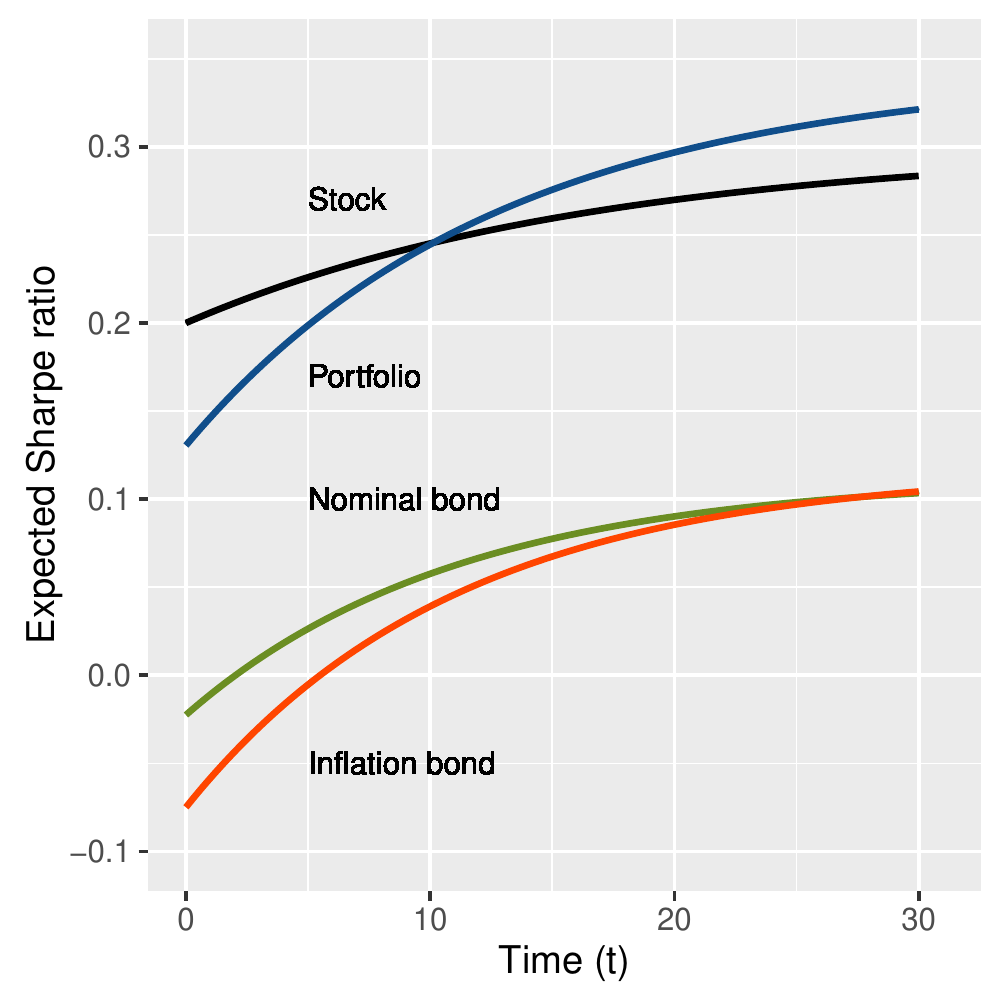}   
\hfill
\includegraphics[height=7.5cm]{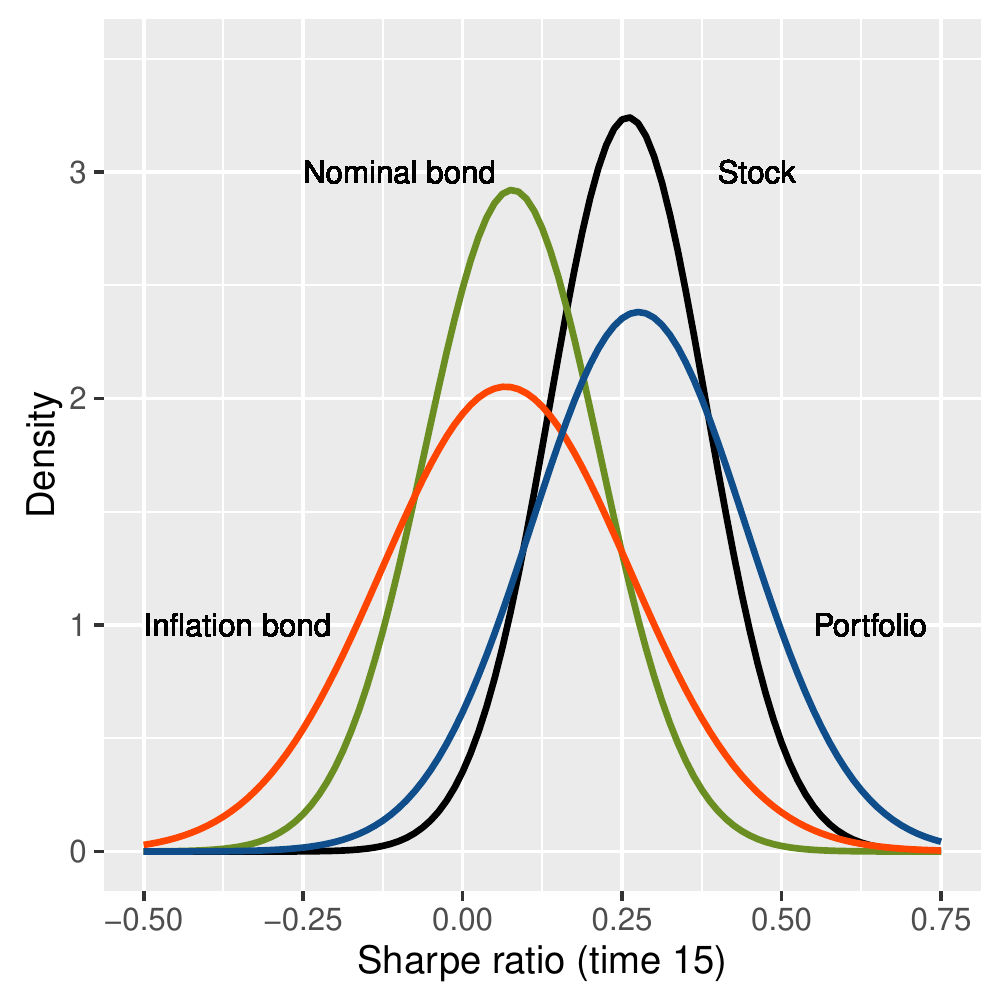}   
\end{center}
\vspace*{-5mm}
\caption{Left plot shows the expected Sharpe ratio for stocks, a nominal bond index (10 years), an inflation bond index (15 years) and a portfolio with factor risk allocation $(f^E,f^R,f^I)=(40,40,20)$.
Right plot shows the Sharpe ratio distribution for the same assets at time 15. Capital market parameters, including initial values, are given in Table~\ref{tab:cappar}.}
\label{fig:SRexample}
\end{figure}

\section{Concluding remarks} \label{sec:concluding}
In this paper we have analysed the five-factor capital market model of \cite{munetal04} with the addition of break-even inflation curves. In particular, we have derived the distributional results which allow for exact simulation from the model. We have also derived a number of auxiliary results regarding constant-maturity bond indices, real value of stocks and portfolio properties.

The application we have in mind is evaluation of different investment strategies via simulation. Consequently, we interpret the model as specifying the P-dynamics ('real world' dynamics) of the state variables. The P-dynamics of nominal and inflation bonds follow when, in addition, we specify market prices of risk and impose absence of arbitrage. Whether we view the model as a P- or Q-model is, however, merely a matter of viewpoint and how we interpret the parameters. If so desired, we could alternatively think of the model as specifying the Q-dynamics (pricing dynamics) and use it for pricing purposes instead.

We have stated a number of results which can be used to illustrate and explore model properties, e.g.\ the distribution of stocks in both nominal and real terms and the distribution of Sharpe ratios for individual assets and portfolios. We envisage that model parameters are chosen to reflect either long-term investment beliefs, or alternatively to reflect more pessimistic views as part of a stress analysis. In either way, we think of a situation where the parameters are chosen by the user and the model is then used to generate a coherent set of stochastic scenarios.

Alternatively, the model can be calibrated in a more 'objective' way by estimating the parameters from data. If all state variables were observable and observed at a set of discrete points in time, we could use Proposition~\ref{prop:conddist} to write down the likelihood function and estimate the parameters by standard MLE. In practice, however, the state variables are only partly observable, since the excess return on stocks ($x_t$) and the expected rate of inflation ($\pi_t$) are not directly available. There are essentially two ways to proceed. Either to find predictor variables, e.g.\ the dividend-price ratio and census views on inflation, for the unobservable quantities, or to apply a Kalman filter to compute the 'observable' likelihood function to be maximized treating $x$ and $\pi$ as truly unobservable. The latter approach has the added benefit of yielding an estimate of the excess return and expected rate of inflation as part of the estimation process. For general information on state space models and Kalman filtering we refer the interested reader to \cite{har89} or \cite{capetal05}, see also \cite{munetal04} for a description of how to use Kalman filtering to estimate the present model, including specific parameter estimates obtained from US data from March 1951 to January 2003.

\appendix
\section{Computation of variances and covariances}
\label{app:computation}
Some care needs to be taken in order to handle the full parameter space, specifically vanishing mean reversion parameters. To state formulas valid for all parameter values we make use of a number of auxiliary functions. After some preliminary definitions, we calculate the variance terms and the first two rows/columns of the variance-covariance matrix given in Table~\ref{tab:moments} on page~\pageref{tab:moments}. The remaining terms of the matrix follow the same pattern and will not be explicitly derived.

For ease of reference in the following we repeat the definition of $\Psi$ from equation (\ref{eq:psidef})
\begin{align}
  \Psi(\kappa,t) \equiv \int_0^t e^{-\kappa u}du =
  \begin{cases}
  t & \mbox{for } \kappa = 0, \\
  \frac{1}{\kappa}\left(1-e^{-\kappa t} \right) & \mbox{for } \kappa \neq 0. \\
  \end{cases}
\end{align}

\subsection{Preliminaries}  \label{sec:prelim}
First, let us recall the Taylor series expansion of $\kappa \rightarrow e^{-\kappa t}$ at $\kappa=0$
\begin{align}
   e^{-\kappa t} = 1 -  \kappa t + \frac{1}{2}t^2\kappa^2 - \frac{1}{6}t^3 \kappa^3 + o(\kappa^3),
\end{align}
where the last term denotes a function with the property that $o(\kappa^3)/\kappa^3 \rightarrow 0$ as $\kappa$ tends to $0$. We can use this expression to show the following limit results in $\kappa$ (for fixed $t$)
\begin{align*}
   \frac{\Psi(\kappa,t)-t}{\kappa} & = \frac{1}{\kappa}\left[\frac{1}{\kappa}\left(1-e^{-\kappa t} \right) - t \right] \\[2mm]
       & = \frac{1}{\kappa}\left[t - \frac{1}{2}t^2\kappa + o(\kappa) - t \right] \to -\frac{t^2}{2} \mbox{ for }\kappa \to 0,
\end{align*}
and
\begin{align*}
  \frac{2}{\kappa}\left[\frac{\Psi(2\kappa,t)-t}{2\kappa} - \frac{\Psi(\kappa,t)-t}{\kappa} \right]
  = \frac{2}{\kappa}\left[\frac{1}{6}t^3(2\kappa) - \frac{1}{6}t^3\kappa + o(\kappa)\right] \to \frac{t^3}{3} \mbox{ for }\kappa \to 0.
\end{align*}

Second, it turns out that many of the variance and covariance expressions take the form of differentials, and differences between differentials, of $\Psi$.
With slight abuse of notation we therefore define the following two functions
\begin{align}
  \Theta(\kappa,t) \equiv \frac{t-\Psi(\kappa,t)}{\kappa} =
  \begin{cases}
  \frac{t^2}{2} & \mbox{for } \kappa = 0, \\
  \frac{1}{\kappa^2}\left(-1 + \kappa t + e^{-\kappa t}\right) & \mbox{for } \kappa \neq 0, \label{eq:Thetadef} \\
  \end{cases}
\end{align}
and
\begin{align}
  \Upsilon(\kappa,t) \equiv \frac{2}{\kappa}\left[ \Theta(\kappa,t)-\Theta(2\kappa,t) \right]
   =
  \begin{cases}
  \frac{t^3}{3} & \mbox{for } \kappa = 0, \\
  \frac{1}{2\kappa^3}\left(-3 + 2\kappa t + 4e^{-\kappa t} - e^{-2\kappa t} \right) & \mbox{for } \kappa \neq 0. \label{eq:Upsilondef} \\
  \end{cases}
\end{align}
Strictly speaking the defining equality in (\ref{eq:Thetadef}) and (\ref{eq:Upsilondef}) is valid only for $\kappa \neq 0$, but the limit results show that
$\Theta$ and $\Upsilon$ might be extended by continuity at $\kappa=0$ with the stated value.

With the same slight abuse of notation we also define the following continuous functions
\begin{align}
  \Gamma(\kappa,\alpha,t) \equiv \frac{\Psi(\alpha,t)-\Psi(\alpha+\kappa,t)}{\kappa}
   =
  \begin{cases}
  \frac{t^2}{2}  & \mbox{for } \kappa = \alpha = 0, \\
  \frac{1}{\alpha^2}\left(1-e^{-\alpha t} - \alpha t e^{-\alpha t}\right)   & \mbox{for } \kappa=0, \alpha \neq 0, \\
  \frac{1}{\kappa}\left(\Psi(\alpha,t)-\Psi(\alpha+\kappa,t) \right) & \mbox{for } \kappa \neq 0. \label{eq:Gammadef} \\
  \end{cases}
\end{align}
The middle case follows by noting that for $\alpha \neq 0$ the ratio converges to $-\partial\Psi/\partial\kappa_{|\kappa=\alpha}$ for $\kappa$ tending to $0$. Also note, to ensure that the expression is valid also for the special cases  $\alpha=0$ and $\alpha+\kappa=0$ we do not expand $\Psi$ in the last case of (\ref{eq:Gammadef}).  The following special cases relates $\Gamma$ to the functions already defined, $\Gamma(\kappa,0,t) = \Theta(\kappa,t)$ and $\Gamma(\kappa,\kappa,t)= \frac{1}{2}\Psi^2(\kappa,t)$.

Finally, we define
\begin{align}
  \Lambda(\kappa,\alpha,t) & \equiv \frac{\Psi(\alpha+\kappa,t)-\Psi(\alpha,t)-\Psi(\kappa,t)+t}{\alpha\kappa} \nonumber \\[2mm]
  & =
  \begin{cases}
  \frac{t^3}{3}  & \mbox{for } \kappa = \alpha = 0, \\
  \frac{1}{2\alpha^3}\left(-2+t^2\alpha^2 + 2e^{-\alpha t} + 2 \alpha t e^{-\alpha t}\right)   & \mbox{for } \kappa=0, \alpha \neq 0, \\
  \frac{1}{2\kappa^3}\left(-2+t^2\kappa^2 + 2e^{-\kappa t} + 2 \kappa t e^{-\kappa t}\right)   & \mbox{for } \kappa \neq 0, \alpha = 0, \\
  \frac{1}{\alpha\kappa}\left(\Psi(\alpha+\kappa,t)-\Psi(\alpha,t)-\Psi(\kappa,t)+t \right) & \mbox{for }   \kappa \neq 0, \alpha \neq 0. \label{eq:Lambdadef} \\
  \end{cases}
\end{align}
Note that $\Lambda$ is symmetric in $\alpha$ and $\kappa$, i.e.\ $\Lambda(\kappa,\alpha,t)=\Lambda(\alpha,\kappa,t)$. The second case follows by noting that when $\alpha\neq 0$ the ratio converges to $(\Theta(0,t) + \partial\Psi/\partial\kappa_{|\kappa=\alpha})/\alpha$ for $\kappa$ tending to $0$. The third case follows by symmetry. To ensure that the expression is valid also for the special case $\alpha+\kappa=0$ we do not expand $\Psi$ in the last case of (\ref{eq:Lambdadef}). Also note that $\Lambda$ is a generalisation of $\Upsilon$ in the sense that $\Lambda(\kappa,\kappa,t)=\Upsilon(\kappa,t)$.

Now, with these definitions we have for all $\alpha$ and $\kappa$
\begin{align}
  \frac{\partial}{\partial t}\Theta(\kappa,t) = \Psi(\kappa,t), & \quad \frac{\partial}{\partial t}\Upsilon(\kappa,t) = \Psi^2(\kappa,t),  \\
    \frac{\partial}{\partial t}\Gamma(\kappa,\alpha,t) = e^{-\alpha t}\Psi(\kappa,t), & \quad \frac{\partial}{\partial t}\Lambda(\kappa,\alpha,t) = \Psi(\alpha,t)\Psi(\kappa,t), \quad
\end{align}
which, together with $\Theta(\kappa,0) =  \Upsilon(\kappa,0) =  \Gamma(\kappa,\alpha,0) =  \Lambda(\kappa,\alpha,0) = 0$, shows
\begin{align}
  \int_0^t \Psi(\kappa,s)ds & = \Theta(\kappa,t),   \label{eq:intpsi} \\[2mm]
  \int_0^t \Psi^2(\kappa,s)ds & = \Upsilon(\kappa,t), \label{eq:intpsi2}\\[2mm]
  \int_0^t e^{-\alpha s}\Psi(\kappa,s)ds & = \Gamma(\kappa,\alpha,t), \label{eq:intexppsi}\\[2mm]
  \int_0^t \Psi(\alpha,s)\Psi(\kappa,s)ds & = \Lambda(\kappa,\alpha,t), \label{eq:intpsipsi}
\end{align}

\subsection{Variances}
From (\ref{eq:rsol}) we have
\begin{align}
  \Var\left[r_t\right] & = \Var\left[\bar{r} + e^{-\kappa t}(r_0 - \bar{r}) + \sigma_r \int_0^t e^{-\kappa(t-s)}dW^r_s\right] \nonumber \\[2mm]
  & = \sigma_r^2\Var\left[\int_0^t e^{-\kappa(t-s)}dW^r_s\right]  \nonumber \\[2mm]
  & = \sigma_r^2 \int_0^t e^{-2\kappa(t-s)}ds \nonumber \\[2mm]
  & = \sigma_r^2 \int_0^t e^{-2\kappa s}ds \nonumber \\[2mm]
  & = \sigma_r^2 \Psi(2\kappa, t), \label{eq:V11}
\end{align}
and further from (\ref{eq:intrsol}) and (\ref{eq:intpsi2})
\begin{align}
  \Var\left[\int_0^t r_s ds\right] & = \Var\left[t \bar{r} + (r_0-\bar{r})\Psi(\kappa,t) + \sigma_r \int_0^t \Psi(\kappa,t-s) dW^r_s\right] \nonumber \\[2mm]
  & = \sigma_r^2 \Var\left[ \int_0^t \Psi(\kappa,t-s) dW^r_s\right] \nonumber \\[2mm]
  & = \sigma_r^2 \int_0^t \Psi^2(\kappa,t-s) ds \nonumber \\[2mm]
  & = \sigma_r^2 \int_0^t \Psi^2(\kappa,s) ds \nonumber \\[2mm]
  & = \sigma_r^2 \Upsilon(\kappa,t). \label{eq:V22}
\end{align}
Similarly, we find from (\ref{eq:xsol}), (\ref{eq:pisol}), (\ref{eq:intxsol}), and (\ref{eq:intpisol})
\begin{align}
  \Var\left[x_t\right]              & = \sigma_x^2 \Var\left[\int_0^t e^{-\alpha(t-s)}dW^S_s \right] = \sigma_x^2 \Psi(2\alpha, t),      \label{eq:V33} \\[2mm]
  \Var\left[\int_0^t x_s ds\right]  & = \sigma_x^2 \Var\left[\int_0^t \Psi(\alpha,t-s) dW^S_s\right] = \sigma_x^2 \Upsilon(\alpha,t),    \label{eq:V44} \\[2mm]
  \Var\left[\pi_t\right]              & = \sigma_\pi^2\Var\left[ \int_0^t e^{-\beta(t-s)}dW^\pi_s\right] = \sigma_\pi^2 \Psi(2\beta, t),   \label{eq:V55} \\[2mm]
  \Var\left[\int_0^t \pi_s ds\right]  & = \sigma_\pi^2 \Var\left[ \int_0^t \Psi(\beta,t-s) dW^\pi_s\right] = \sigma_\pi^2 \Upsilon(\beta,t). \label{eq:V66}
\end{align}
Finally, we state for completeness
\begin{align}
  \Var\left[W^S_t\right]       = t.      \label{eq:V77}
\end{align}

\subsection{Covariances involving $r_t$}
From (\ref{eq:rsol}) and (\ref{eq:intrsol}) we find
\begin{align}
  \Cov\left[r_t,\int_0^t r_sds \right]
   & = \Cov\left[\sigma_r \int_0^t e^{-\kappa(t-s)}dW^r_s,\sigma_r \int_0^t \Psi(\kappa,t-s) dW^r_s \right] \nonumber \\[2mm]
   & = \sigma_r^2 \int_0^t e^{-\kappa(t-s)}\Psi(\kappa,t-s) ds \nonumber \\[2mm]
   & = \sigma_r^2 \int_0^t e^{-\kappa s}\Psi(\kappa,s) ds \nonumber \\[2mm]
   & = \sigma_r^2 \Gamma(\kappa,\kappa,t),   \label{eq:V12}
\end{align}
where the last equality uses (\ref{eq:intexppsi}).

From (\ref{eq:rsol}) and (\ref{eq:xsol}) we get
\begin{align}
  \Cov\left[r_t, x_t \right]
  & = \Cov\left[\sigma_r \int_0^t e^{-\kappa(t-s)}dW^r_s, -\sigma_x \int_0^t e^{-\alpha(t-s)}dW^S_s\right] \nonumber \\[2mm]
  & = -\sigma_r \sigma_x \rho_{rS} \int_0^t e^{-(\alpha+\kappa)(t-s)}ds \nonumber \\[2mm]
& = -\sigma_r \sigma_x \rho_{rS} \int_0^t e^{-(\alpha+\kappa)s} ds \nonumber \\[2mm]
& = -\sigma_r \sigma_x \rho_{rS} \Psi(\alpha+\kappa,t),  \label{eq:V13}
\end{align}
where $\rho_{rS}$ denotes the correlation between $W^r$ and $W^S$.

From (\ref{eq:rsol}), (\ref{eq:intxsol}) we get
\begin{align}
  \Cov\left[r_t, \int_0^t x_s ds \right]
  & = \Cov\left[\sigma_r \int_0^t e^{-\kappa(t-s)}dW^r_s, - \sigma_x \int_0^t \Psi(\alpha,t-s) dW^S_s\right] \nonumber \\[2mm]
  & = -\sigma_r \sigma_x \rho_{rS} \int_0^t e^{-\kappa(t-s)}\Psi(\alpha,t-s)ds \nonumber \\[2mm]
  & = -\sigma_r \sigma_x \rho_{rS} \int_0^t e^{-\kappa s}\Psi(\alpha,s)ds \nonumber \\[2mm]
  & = -\sigma_r \sigma_x \rho_{rS} \Gamma(\alpha,\kappa,t). \label{eq:V14}
\end{align}

Similarly, we get using (\ref{eq:pisol}) and (\ref{eq:intpisol})
\begin{align}
  \Cov\left[r_t, \pi_t \right] & = \sigma_r \sigma_\pi \rho_{r\pi} \Psi(\beta+\kappa,t),  \label{eq:V15} \\[2mm]
  \Cov\left[r_t, \int_0^t \pi_s ds \right] & = \sigma_r \sigma_\pi \rho_{r\pi} \Gamma(\beta,\kappa,t).  \label{eq:V16}
\end{align}
The covariance between $r_t$ and $W^S_t$ is given by
\begin{align}
  \Cov\left[r_t, W^S_t \right]
   = \Cov\left[\sigma_r \int_0^t e^{-\kappa(t-s)}dW^r_s, \int_0^t 1 dW^S_s\right] = \sigma_r \rho_{rS} \Psi(\kappa,t).  \label{eq:V17}
\end{align}

\subsection{Covariances involving $\int_0^t r_s ds$}
From (\ref{eq:intrsol}) and (\ref{eq:xsol}) we get
\begin{align}
  \Cov\left[\int_0^t r_sds, x_t \right]
   & = \Cov\left[\sigma_r \int_0^t \Psi(\kappa,t-s) dW^r_s,-\sigma_x \int_0^t e^{-\alpha(t-s)}dW^S_s \right] \nonumber \\[2mm]
  & = -\sigma_r \sigma_x \rho_{rS} \int_0^t e^{-\alpha(t-s)}\Psi(\kappa,t-s)ds \nonumber \\[2mm]
  & = -\sigma_r \sigma_x \rho_{rS} \int_0^t e^{-\alpha s}\Psi(\kappa,s)ds \nonumber \\[2mm]
  & = -\sigma_r \sigma_x \rho_{rS} \Gamma(\kappa,\alpha,t), \label{eq:V23}
\end{align}
and from (\ref{eq:intrsol}) and (\ref{eq:intxsol}) we get
\begin{align}
  \Cov\left[\int_0^t r_sds, \int_0^t x_sds \right]
   & = \Cov\left[\sigma_r \int_0^t \Psi(\kappa,t-s) dW^r_s,- \sigma_x \int_0^t \Psi(\alpha,t-s) dW^S_s \right] \nonumber \\[2mm]
  & = -\sigma_r \sigma_x \rho_{rS} \int_0^t \Psi(\alpha,t-s)\Psi(\kappa,t-s)ds \nonumber \\[2mm]
  & = -\sigma_r \sigma_x \rho_{rS} \int_0^t \Psi(\alpha,s)\Psi(\kappa,s)ds \nonumber \\[2mm]
  & = -\sigma_r \sigma_x \rho_{rS} \Lambda(\kappa,\alpha,t). \label{eq:V24}
\end{align}
Similarly, we get using (\ref{eq:pisol}) and (\ref{eq:intpisol})
\begin{align}
  \Cov\left[\int_0^t r_sds, \pi_t \right] & = \sigma_r \sigma_\pi \rho_{r\pi} \Gamma(\kappa,\beta,t),  \label{eq:V25} \\[2mm]
  \Cov\left[\int_0^t r_sds, \int_0^t \pi_s ds \right] & = \sigma_r \sigma_\pi \rho_{r\pi} \Lambda(\kappa,\beta,t), \label{eq:V26}
\end{align}
while the covariance between $\int_0^t r_s ds$ and $W^S_t$ is given by
\begin{align}
  \Cov\left[\int_0^t r_sds, W^S_t \right]
   = \Cov\left[\sigma_r \int_0^t \Psi(\kappa,t-s)dW^r_s, \int_0^t 1 dW^S_s\right] = \sigma_r \rho_{rS} \Theta(\kappa,t).  \label{eq:V27}
\end{align}

\pagebreak
\bibliographystyle{plainnat}
\bibliography{capmarket}

\end{document}